%% file: SharpRec.tex
\documentclass[sigconf]{acmart}
\AtBeginDocument{%
  }

\copyrightyear{2026}
\acmYear{2026}
\setcopyright{cc}
\setcctype{by}
\acmConference[KDD '26]{Proceedings of the 32nd ACM SIGKDD Conference on Knowledge Discovery and Data Mining V.2}{August 09--13, 2026}{Jeju Island, Republic of Korea}
\acmBooktitle{Proceedings of the 32nd ACM SIGKDD Conference on Knowledge Discovery and Data Mining V.2 (KDD '26), August 09--13, 2026, Jeju Island, Republic of Korea}
\acmDOI{10.1145/3770855.3817945}
\acmISBN{979-8-4007-2259-2/2026/08}

\usepackage{enumitem}
\usepackage{amsmath}
\usepackage{subfigure}
\usepackage[table]{xcolor}
\usepackage{multirow}
\usepackage{bm}

\newcommand{\method}{\textsf{SharpRec}}

\begin{document}

\title[Sharpness-aware Model Merging with Salience Recovery for LLM-based \\ Cross-Domain Sequential Recommendation]{Sharpness-aware Model Merging with Salience Recovery for LLM-based Cross-Domain Sequential Recommendation}


\author{Huwei Ji}
\orcid{0009-0009-7273-0372}
\affiliation{%
  \institution{Zhejiang University}
  \city{Hangzhou}
  \country{China}
}
\email{jihuwei@zju.edu.cn}

\author{Jiajie Su}
\orcid{0000-0002-6899-4174}
\affiliation{%
  \institution{Zhejiang University}
  \city{Hangzhou}
  \country{China}
}
\email{sujiajie@zju.edu.cn}

\author{Yuyuan Li}
\orcid{0000-0003-4896-2885}
\affiliation{%
  \institution{Hangzhou Dianzi University}
  \city{Hangzhou}
  \country{China}
}
\email{y2li@hdu.edu.cn}

\author{Xiaohua Feng}
\orcid{0009-0001-6829-7088}
\affiliation{%
  \institution{Zhejiang University}
  \city{Hangzhou}
  \country{China}
}
\email{fengxiaohua@zju.edu.cn}

\author{Chaochao Chen}
\authornote{Corresponding Author.}
\orcid{0000-0003-1419-964X}
\affiliation{%
  \institution{Zhejiang University}
  \city{Hangzhou}
  \country{China}
}
\email{zjuccc@zju.edu.cn}

\renewcommand{\shortauthors}{Huwei Ji, Jiajie Su, Yuyuan Li, Xiaohua Feng, \& Chaochao Chen}

\begin{abstract}
LLM-based Cross-Domain Sequential Recommendation (CDSR) leverages LLMs to enhance target performance via deep semantic reasoning, alleviating the dependency on overlapping users. Among LLM-based paradigms, model merging is particularly promising for multi-domain scenarios due to its superior scalability and flexibility in integrating diverse knowledge sources.
However, our empirical investigations reveal two critical bottlenecks: (1) cross-domain knowledge conflict; 
and (2) performance saturation in multi-domain fusion.
Our analysis attributes these phenomena to parameter-level misalignment and statistical homogenization during the merging process. 
To address these bottlenecks, we propose SharpRec, Sharpness-aware Model Merging with Salience Recovery for LLM-based CDSR, a framework designed to lift the performance upper bound of merged models. 
SharpRec incorporates two synergistic modules: Sharpness-aware Geometric Alignment to establish a stable geometric foundation for interference-free fusion; and Preference Salience Activation to effectively recover the distinctive features essential for bolstering target domain performance.
Extensive experiments in both dual-domain and multi-domain scenarios demonstrate that SharpRec consistently outperforms state-of-the-art baselines.
\end{abstract}





\maketitle

\input{body/1-intro}

\input{body/2-re}
\input{body/3-pre}

\input{body/4-met}
\input{body/5-exp}
\input{body/6-con}


\begin{acks}
This work was supported in part by the National Natural Science Foundation of China (No.~62522217, No.~62402148) and the Zhejiang Provincial Natural Science Foundation of China (No.~LZYQ25F020002).
\end{acks}

\bibliographystyle{ACM-Reference-Format}
\balance
\bibliography{sample-base}

\appendix

\input{body/app}

\end{document}

%% file: body/1-intro.tex
\section{Introduction}
\begin{figure*}[t] 
    \centering

    \subfigure[Performance Disparity]{\includegraphics[width=0.245\linewidth]{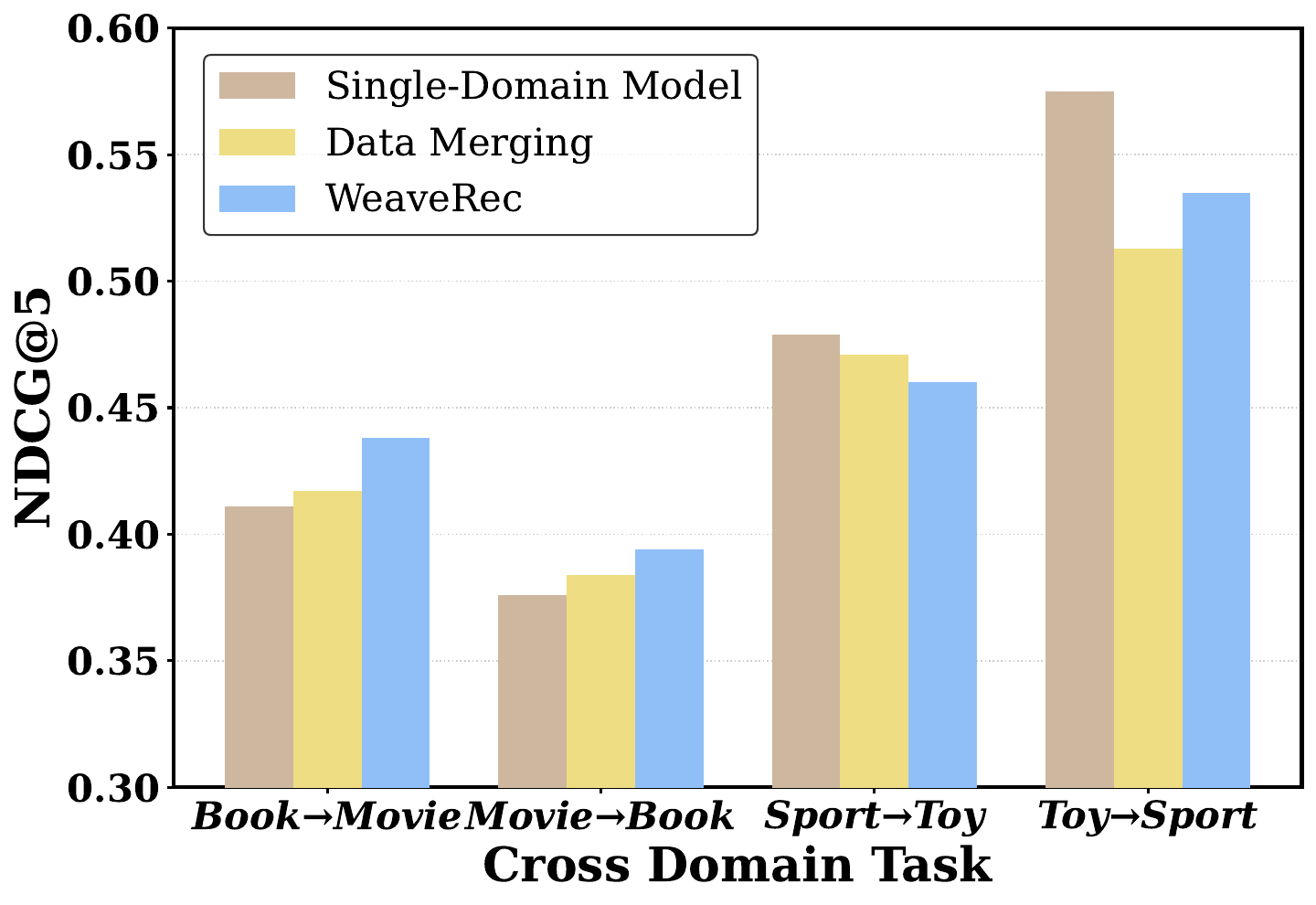}}
    \hfill %
    \subfigure[Weight Tuning Failure]{\includegraphics[width=0.245\linewidth]{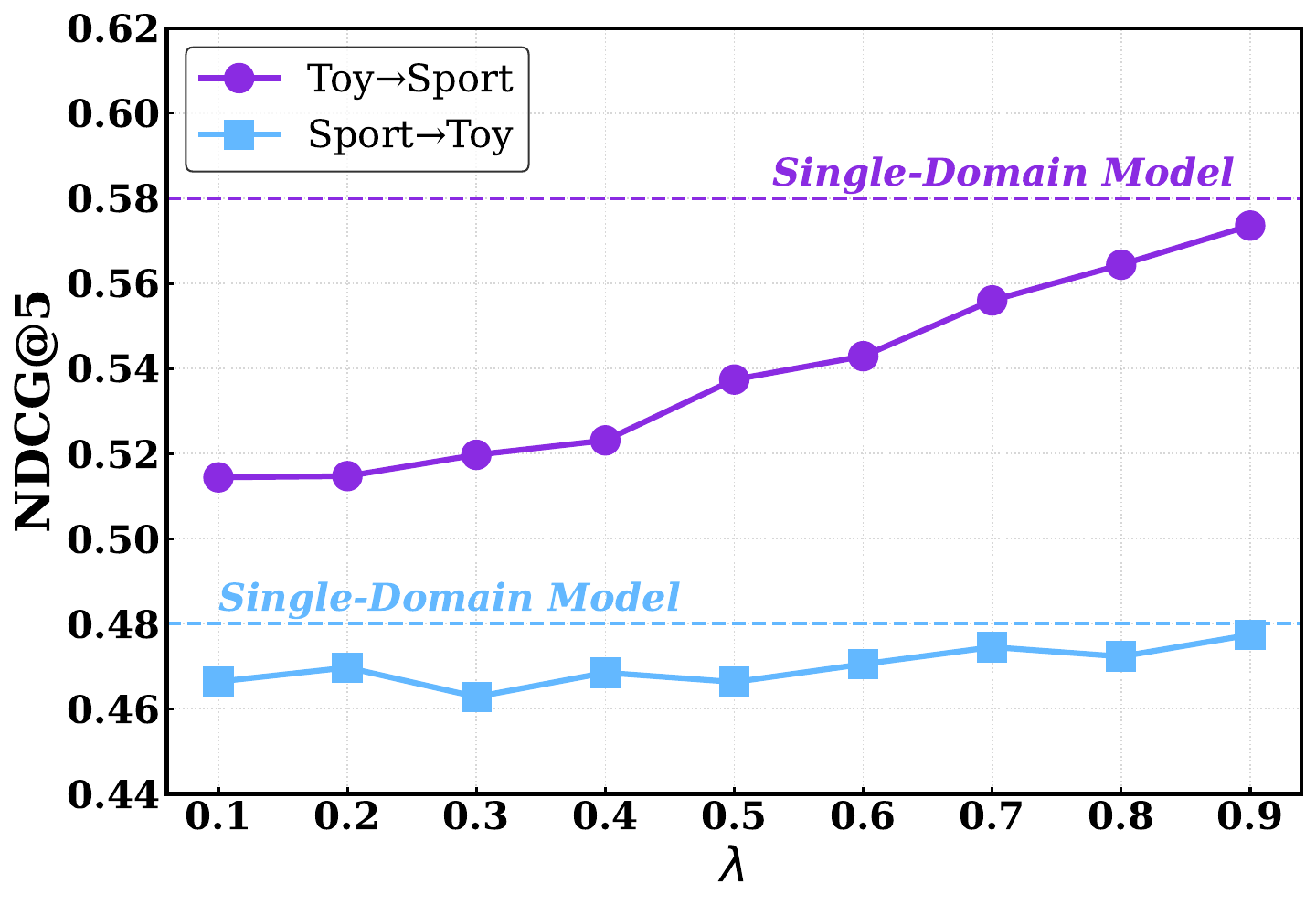}}
    \hfill
    \subfigure[Performance Saturation]{\includegraphics[width=0.245\linewidth]{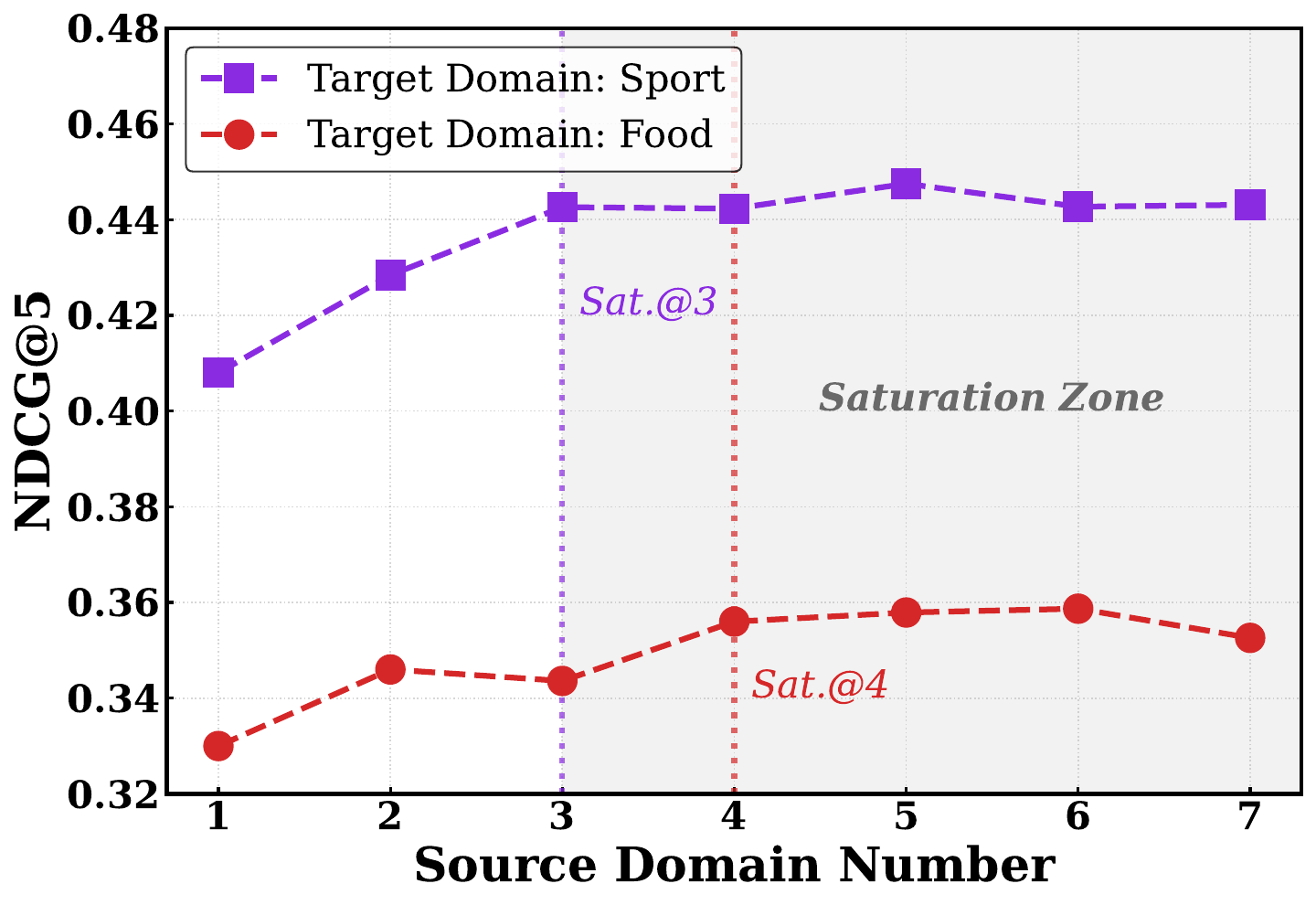}}
    \hfill %
    \subfigure[Statistical Homogenization]{\includegraphics[width=0.245\linewidth]{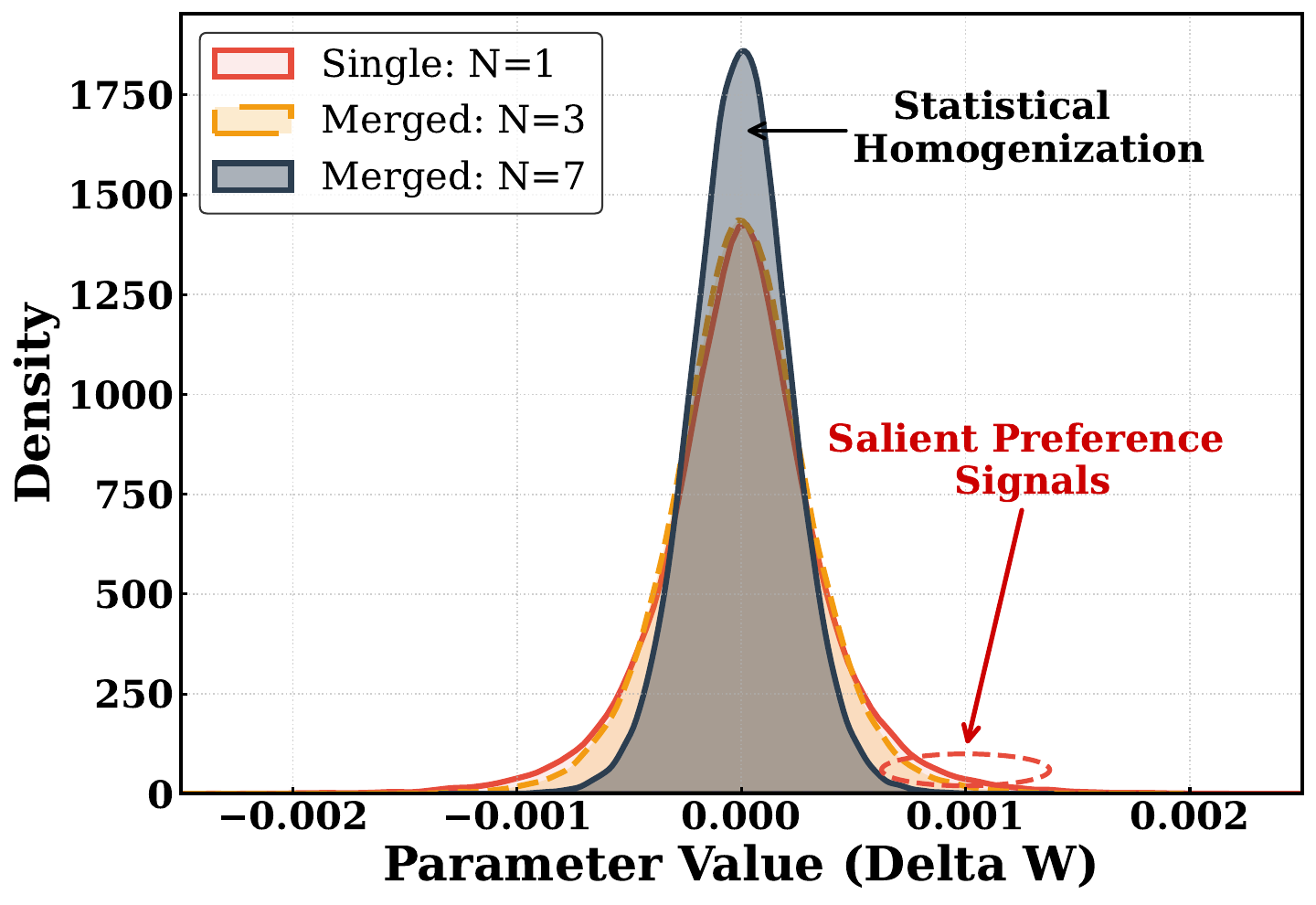}}


    \caption{Empirical evidence of negative transfer and scalability bottlenecks in CDSR. (a) \& (b) show that merging heterogeneous domains causes negative transfer which cannot be resolved by weight tuning. (c) \& (d) reveal that multi-domain fusion leads to performance saturation due to statistical parameter homogenization.}
    \label{fig:intro}
\end{figure*}
Cross-Domain Sequential Recommendation (CDSR) has attracted significant attention for its capability to enhance recommendation performance in the target domain by transferring knowledge from rich auxiliary domain interactions.
Most early CDSR approaches~\citep{ma2019pi, cao2022contrastive} operate under the strict assumption of fully overlapping users, relying on them to align cross-domain feature spaces. 
But real-world applications typically present a \textit{hybrid user scenario}, which consists of a small fraction of overlapping users and a vast majority of non-overlapping users active in only a single domain.
%
%
%
To tackle the hybrid user scenario, recent studies~\citep{xu2024rethinking, xu2024towards, li2024cross, feng2026generalized} have employed advanced mechanisms, such as heterogeneous graphs or meta-learning, to mine latent cross-domain correlations. 
A key limitation is that these approaches predominantly rely on ID-based representations that lack deep semantic understanding, making them ineffective at bridging domain gaps for non-overlapping users where explicit co-occurrence signals are scarce.

To bridge this semantic gap, emerging research integrates LLMs to empower CDSR with rich open-world knowledge. 
By leveraging their robust reasoning capabilities, LLMs can infer implicit item correlations and uncover shared cross-domain features, effectively aligning domains for non-overlapping users even without interaction links.
%
%
Current research integrating LLMs into CDSR can be broadly categorized into \textit{LLM-enhanced} approaches~\citep{liu2025bridge, shen2024exploring, xin2025llmcdsr, wang2025lecdsr} and \textit{LLM-based} approaches~\citep{cui2022m6, geng2022recommendation, liu2025uncovering, peng2024ecellm, tang2025one}.
LLM-enhanced methods focus on aligning cross-domain feature distributions by projecting items into a unified semantic space via adapters or contrastive learning. 
However, this paradigm suffers from semantic-collaborative misalignment, where forced spatial mapping leads to information distortion and negative transfer in dynamic preference modeling.
LLM-based approaches aggregate multi-domain recommendation data into unified instruction-tuning datasets to train a single comprehensive model under a one model paradigm.
%
While effective, this method often leads to data distribution conflicts across domains in real-world scenarios, and lacks flexibility, as domain updates necessitate costly full-model retraining.
%
%

To overcome these inherent limitations, \textit{Model Merging}~\citep{yang2408model, pmlr-v267-tan25f} has emerged as a more advantageous paradigm.
By directly integrating domain-specific models (e.g., LoRA adapters) in the parameter space, such approaches enable cross-domain knowledge transfer in a flexible and efficient manner, without enforcing representation-space alignment or requiring retraining.
%
%
For instance, X-Cross~\citep{hadad2025x} employs a dynamic integration mechanism to linearly merge multiple expert models at the layer level, while WeaveRec~\citep{hou2025weaverec} applies a linear weight average to merge LoRA parameters. 
Despite the promise of model merging in enabling flexible knowledge transfer, its direct application in realistic CDSR scenarios is constrained by two bottlenecks, as revealed by our preliminary investigation (detailed in Sec. ~\ref{sec:motivation}).
\begin{itemize}[leftmargin=*]\setlength{\itemsep}{-\itemsep}
\item \textbf{B1: Cross-Domain Knowledge Conflict.}
Merging heterogeneous domains inevitably triggers \textit{parameter interference}, leading to severe cross-domain knowledge conflict.
As illustrated in Figure~\ref{fig:intro}(a), while compatible pairs (e.g., \textit{Book-Movie}) show gains, merging divergent domains (e.g., \textit{Sport-Toy}) yields performance inferior to the single-domain baseline.
Figure~\ref{fig:intro}(b) further confirms that even through exhaustive optimization on the linear merging weight $\lambda$, the cross-domain performance still fails to surpass the single-domain model. 
Our empirical analysis (detailed in Sec. ~\ref{sec:motivation}, Figure~\ref{fig:para_vis}) finds that this failure stems from geometric incompatibility, where independent fine-tuning traps models in conflicting sharp minima.

\item \textbf{B2: Rapid Saturation in Multi-domain Fusion.}
Scaling up source domains triggers \textit{statistical parameter homogenization}, as linear aggregation dilutes salient coefficients, causing rapid performance saturation.
As illustrated in Figure~\ref{fig:intro}(c), performance gains on targets like \textit{Sport} and \textit{Food} exhibit diminishing returns, hitting a distinct plateau after integrating 3 or 4 source domains. 
Figure~\ref{fig:intro}(d) reveals that this saturation stems from linear aggregation functioning as a mean filter, which systematically smooths out distributional outliers into a generic, feature-poor Gaussian distribution.
Consequently, this erosion of parameter saliency destroys the high-order signals essential for capturing subtle cross-domain preference correlations in CDSR tasks.

%
\end{itemize}
Taken together, these observations suggest that the effectiveness of existing model merging methods in CDSR is fundamentally constrained by the \textbf{geometric parameter interference and statistical homogenization of preference signals.}

To address the fundamental bottlenecks of geometric incompatibility and statistical homogenization, we propose \method{}, a novel framework designed to break the scalability ceiling of LLM-based CDSR. 
The core insight of \method{} lies in the synergy between geometric alignment for interference-free knowledge fusion and distributional reshaping for salient signal reactivation.
Specifically, our framework consists of two collaborative modules: 
(1) \textit{Sharpness-aware Geometric Alignment (SGA)}. 
To eliminate the root cause of parameter interference (\textbf{B1}), we introduce a sharpness-aware tuning mechanism during the domain-specific fine-tuning. 
By guiding models toward flat minima rather than sharp ones, SGA ensures that diverse domain models reside in geometrically connected low-loss basins, thereby securing a stable geometric foundation for the seamless transfer of user preference structures. 
(2) \textit{Preference Salience Activation (PSA)}. 
To overcome the performance ceiling imposed by statistical homogenization (\textbf{B2}), we propose a post-fusion non-linear reparameterization strategy.
Instead of converging to the mediocre Gaussian distribution induced by linear aggregation, PSA reconstructs the merged parameters to restore heavy-tailed characteristics.
By reactivating the salient preference signals eroded during fusion, this mechanism effectively captures high-order user heterogeneity, amplifying the model's capacity to bolster target domain performance.
%

The main contributions are summarized: 
(1) We propose a novel and efficient LLM-based CDSR framework \method{} which introduces a new paradigm by integrating sharpness-aware optimization and distributional recovery for robust and scalable knowledge transfer. 
(2) To address the challenges, we propose: i) \textbf{SGA} guides models into flat minima to resolve geometric incompatibility, thereby mitigating cross-domain parameter interference; ii) \textbf{PSA} reactivates salient preference signals through non-linear reparameterization, lifting the performance upper bound and alleviating the saturation effect in multi-domain fusion. 
(3) Extensive experiments in both dual-domain and multi-domain
scenarios comprehensively validate the effectiveness and scalability of \method{}.

%% file: body/2-re.tex
\section{Related Works}

\subsection{Cross-Domain Sequential Recommendation}
Cross-Domain Sequential Recommendation (CDSR)~\citep{chen2024survey} enhances target-domain performance by exploiting historical interactions across multiple related domains.
By integrating cross-domain recommendation~\citep{gao2023autotransfer, zhang2024m3oe, wang2023plate} with sequential modeling~\citep{liu2023diffusion, li2023strec, zheng2025collabedit}, CDSR frameworks must capture temporal interest evolution while effectively transferring cross-domain preferences.
Early studies~\citep{ma2019pi, cao2022contrastive, xu-etal-2025-videoeraser} primarily facilitate knowledge transfer via overlapping users or items.
For instance, $\pi$-Net~\citep{ma2019pi} and PSJNet~\citep{sun2021parallel} employ RNNs or attention mechanisms to map source-domain representations into target models.
Subsequent research incorporates dual attention~\citep{li2021dual, huang2024troika}, parallel architectures~\citep{ma2024triple, tananimate}, and multi-interest modeling~\citep{ma2022mixed, HABIT, du2021cert} to enrich cross-domain sequential patterns.
To capture finer-grained dependencies, recent methods utilize Graph Neural Networks (GNNs)~\citep{guo2021gcn, xu2025multi, ma2022mixed} and contrastive learning~\citep{zang2023contrastive, xu2024towards,hou2024cross} to model complex interaction structures and cross-domain relationships.
Furthermore, CDSR has expanded to handle realistic constraints, such as partial user overlap~\citep{xu2024towards, li2024cross, fu2025objectrelator} and multi-domain scenarios~\citep{xu2024rethinking, 11494811, gong2020hamming}, using auxiliary behaviors and meta-learning~\citep{li2024cross, TEMA, tan2025mimir} to alleviate data sparsity and enhance scalability.
Despite these advances, most existing CDSR methods remain constrained by limited collaborative signals and shallow representations, failing to capture the complex, dynamic interest transfer inherent in multi-domain environments.

\subsection{LLM for Cross-Domain Sequential Recommendation}
The emergence of Large Language Models (LLMs) has introduced strong semantic reasoning to CDSR. 
Recent studies leverage LLMs to alleviate data sparsity and enhance preference modeling without relying on explicit overlap.
M6-Rec~\citep{cui2022m6} supports open-ended domains by unifying multi-domain data within a generative pretraining framework.
LLM-Rec~\citep{tang2025one} treats cross-domain recommendation as a language modeling task to directly capture user behavior sequences.
URLLM~\citep{shen2024exploring} integrates user retrieval signals into LLMs to improve performance under sparse interaction settings.
LLMCDSR~\citep{xin2025llmcdsr} constructs unified semantic item representations and user profiles to bridge domains.
LeCDSR~\citep{wang2025lecdsr} fuses LLM-generated semantic embeddings with ID-based representations to enhance sequential modeling.
LLM4CDSR~\citep{liu2025bridge} employs hierarchical profiling to capture item relations and global preferences, addressing overlap scarcity.
X-Cross~\citep{hadad2025x} utilizes parameter-efficient fine-tuning and layer-wise integration of domain-specific LLMs for scalable adaptation.
WeaveRec~\citep{hou2025weaverec} investigates model merging by weaving domain-specific LoRA adapters to stabilize knowledge transfer and maintain efficiency.

Despite this progress, fundamental challenges remain: most methods rely on shallow semantic fusion or direct knowledge aggregation without explicitly modeling preference complementarity. This limits their ability to capture dynamic interest transfer in realistic settings.
This motivates us to revisit CDSR from a preference-space modeling perspective, focusing on more effective and scalable knowledge integration.

%% file: body/3-pre.tex
\section{Preliminaries}




\subsection{Problem Formulation of CDSR}

Given a set of $K$ distinct domains denoted as $\mathcal{D} = \{ D_1, D_2, \dots, D_K \}$, each domain $D_k$ is characterized by its own user set $\mathcal{U}_k$ and item set $\mathcal{V}_k$. 
In a realistic hybrid user scenario, user sets may partially overlap ($\mathcal{U}_i \cap \mathcal{U}_j \neq \emptyset$), though most users remain domain-specific.
For a specific user $u \in \mathcal{U}_k$ in domain $D_k$, the interaction history is represented as a chronological sequence $\mathcal{S}_u^k = [v_1, v_2, \dots, v_T]$, where $v_t \in \mathcal{V}_k$ denotes the item interacted with at time step $t$, and $T$ is the current sequence length. 
The objective of CDSR is to accurately predict the next item $v_{T+1}$ by leveraging both the intra-domain sequential patterns and the inter-domain transferable knowledge.

Formally, let $\mathcal{M} = \{ \phi_1, \phi_2, \dots, \phi_K \}$ represent the ensemble of knowledge representations (e.g., model parameters or adapters) derived from all $K$ domains. 
The task is to learn a unified predictive framework that maximizes the following joint log-likelihood across all domains:
\begin{equation}
\max_{\mathcal{M}} \sum_{k=1}^{K} \sum_{u \in \mathcal{U}_k} \log P(v_{T+1} \mid \mathcal{S}_u^k; \mathcal{M}),
\end{equation}
where $P(v_{T+1} \mid \mathcal{S}_u^k; \mathcal{M})$ denotes the probability of the next-item interaction, conditioned on the historical sequence $\mathcal{S}_u^k$ and the cross-domain collaborative knowledge $\mathcal{M}$.

\subsection{LLM-based CDSR}
\label{sec:LLM CDSR}

\noindent \textbf{Generative Modeling for CDSR.} To bridge the semantic gap across $K$ domains, the CDSR task is reformulated within a generative language modeling paradigm. 
Specifically, for a user $u \in \mathcal{U}_k$, the interaction history $\mathcal{S}_u^k$ and the target item $v_{T+1} \in \mathcal{V}_k$ are transformed into textual representations (e.g., item titles or descriptions) via a template-based serialization function $\mathcal{T}(\cdot)$.

The input sequence $\mathcal{S}_u^k$ is encapsulated into a textual prompt $x = \mathcal{T}(\mathcal{S}_u^k)$, while the ground-truth next item is represented as a target string $y = \mathcal{T}(v_{T+1})$. 
This transformation maps domain-specific behaviors into a shared semantic space, enabling the LLM to leverage its inherent open-world knowledge for cross-domain reasoning. 
Given the prompt $x$, the model estimates the probability of generating the target sequence $y$ token by token. 
The instruction-tuning dataset for each domain is defined as $\mathcal{D}_k = \{ (\mathcal{T}(\mathcal{S}_u^k), \mathcal{T}(v_{T+1})) \mid u \in \mathcal{U}_k \}$.

Following the objective defined in Eq. (1), the training process optimizes the parameters $\Phi$ (representing the ensemble $\mathcal{M}$) to maximize the conditional log-likelihood:
\begin{equation}
\max_{\Phi} \sum_{k=1}^{K} \sum_{(x,y) \in \mathcal{D}_k} \sum_{i=1}^{|y|} \log P_{\Phi} (y_i \mid x, y_{<i}),
\end{equation}
where $y_i$ denotes the $i$-th token of the target sequence $y$, and $\mathcal{D}_k$ represents the instruction-tuning dataset derived from domain $D_k$. 
In this generative setup, the knowledge representations $\phi_k \in \mathcal{M}$ are manifested as domain-specific model states or adapters.

\vspace{0.2cm}
\noindent \textbf{Parameter-Efficient Fine-Tuning.} 
In practice, full-parameter fine-tuning of LLMs for each domain is often computationally prohibitive. 
To materialize the domain-specific knowledge ensemble $\mathcal{M} = \{ \phi_1, \dots, \phi_K \}$, Low-Rank Adaptation (LoRA) is commonly employed. 
For a pre-trained weight matrix $W$ within the LLM backbone $\Phi$, LoRA introduces a low-rank update:
\begin{equation}
W' = W + \Delta W = W + \frac{\alpha}{r} BA,
\end{equation}
where $A \in \mathbb{R}^{r \times d}$ and $B \in \mathbb{R}^{d \times r}$ are trainable matrices of rank $r \ll d$. 

In the CDSR setting, the shared backbone $\Phi$ is frozen, while domain-specific LoRA parameters $\theta_k = \{A_k, B_k\}$ are independently optimized for each domain $D_k$. 
Building upon Eq. (2), the objective for acquiring the $k$-th knowledge representation $\phi_k$ is:
\begin{equation}
\max_{\theta_k} \sum_{(x,y) \in \mathcal{D}_k} \sum_{i=1}^{|y|} \log P_{\Phi, \theta_k} (y_i \mid x, y_{<i}),
\end{equation}
where the resulting ensemble $\Theta = \{ \theta_1, \dots, \theta_K \}$ encapsulates modular knowledge from all domains, serving as the prerequisite for subsequent model merging.

%% file: body/4-met.tex
\section{Methodology}

\subsection{Empirical Study}
\label{sec:motivation}
\subsubsection{Model Merging for CDSR}
To integrate knowledge across multiple domains efficiently, model merging is introduced as a flexible paradigm for LLM-based CDSR. 
Building upon the PEFT framework in Sec. 3.2, the process of model merging is typically conducted in two stages:
\begin{itemize}[leftmargin=*]\setlength{\itemsep}{-\itemsep}
    \item \textbf{Domain-Specific Adaptation.} 
    For each domain $D_k \in \mathcal{D}$, a domain-specific knowledge representation $\phi_k$ is first acquired to capture its unique sequential patterns. 
    Following Eq. (4), this involves optimizing the LoRA parameters $\theta_k = \{A_k, B_k\}$ while keeping the backbone $\Phi$ frozen. 
    The resulting effective adapter update $\Delta W_k = \frac{\alpha}{r}B_kA_k$, induced by the LoRA factors $\theta_k=\{A_k,B_k\}$, captures the specialized sequential patterns of domain $D_k$.
    \item \textbf{Parameter-Space Fusion.} 
    Once the ensemble of adapters $\Theta = \{\theta_1, \dots, \theta_K\}$ is obtained, a unified merged model is formed by aggregating these domain-specific updates directly in the parameter space. 
    Formally, for a cross-domain recommendation task, the parameters of the merged model $\Phi_{merge}$ are defined as the combination of the frozen backbone $\Phi$ and the weighted average of all domain-specific adapters:
    \begin{equation}
    \Phi_{merge} = \Phi + \sum_{k=1}^{K}\lambda_k \cdot \Delta W_k = \Phi + \sum_{k=1}^{K}\lambda_k \cdot \left(\frac{\alpha}{r}B_kA_k\right),
    \end{equation}

    where $\lambda_k \in [0, 1]$ represents the merging coefficient (e.g., layer-wise or module-wise) assigned to the $k$-th domain, subject to the constraint $\sum_{k=1}^{K} \lambda_k = 1$. 
    This operation is applied to each LoRA-equipped weight matrix.
    This fusion process allows the model to leverage collaborative signals from both the target domain and multiple auxiliary source domains without the necessity of costly joint retraining.
\end{itemize}

Despite its efficiency, current model merging paradigms for CDSR predominantly operate under the \textbf{ideal assumption of merging a few closely related domains}. 
However, real-world recommendation scenarios often involve diverse domain pairs with inherent semantic conflicts and the necessity of integrating knowledge from a large number of auxiliary sources. 
This discrepancy between ideal assumptions and complex realities significantly limits the practical applicability and scalability of existing methods. 
Through empirical studies, we identify two fundamental bottlenecks that limit the scalability of current methods: (i) Cross-Domain Knowledge Conflict, and (ii) Performance Saturation.





\begin{figure}[t] 
    \centering

    \subfigure[Book → Movie]{\includegraphics[width=0.49\linewidth]{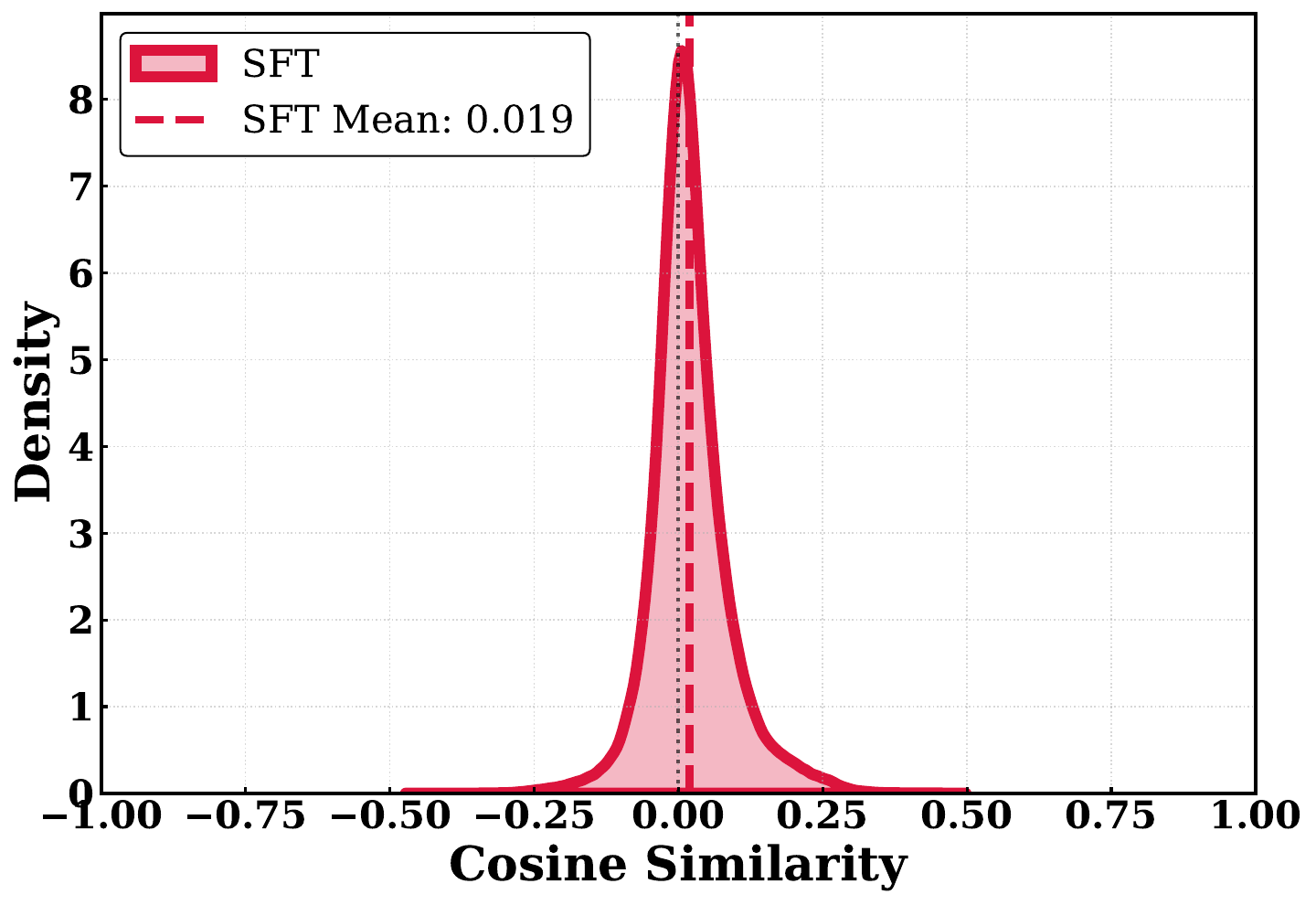}}
    \hfill %
    \subfigure[Sport → Toy]{\includegraphics[width=0.49\linewidth]{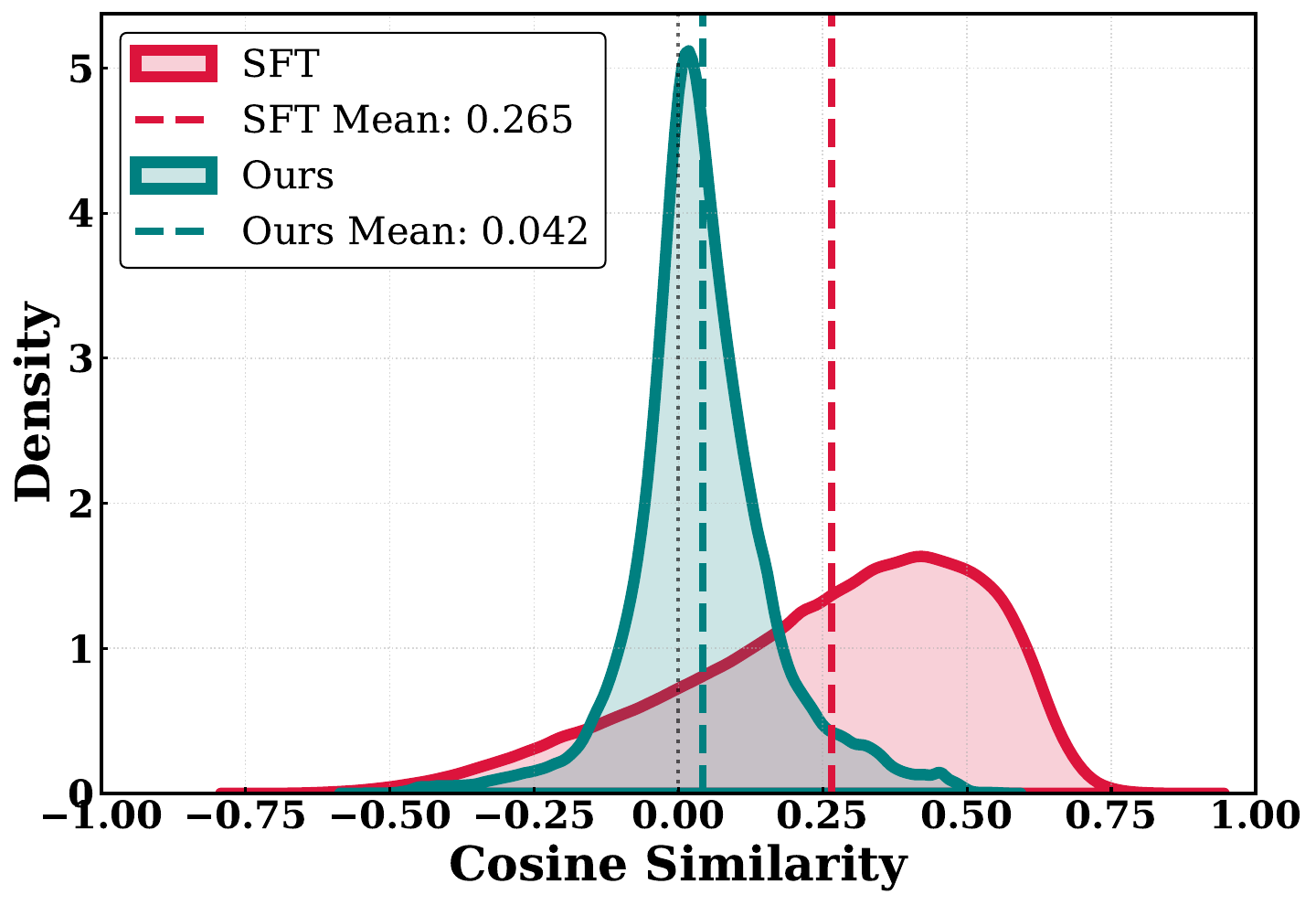}}


    \caption{Visualization of parameter geometric compatibility.}
    \label{fig:para_vis}
\end{figure}

\subsubsection{Cross-Domain Knowledge Conflict}

Our empirical investigation reveals that the efficacy of cross-domain enhancement is heavily contingent on domain compatibility. 
As illustrated in Figure~\ref{fig:intro}(a), merging highly correlated domains yields significant gains, whereas heterogeneous pairs suffer from negative transfer that cannot be resolved even by an exhaustive grid search over fusion weights (Figure~\ref{fig:intro}(b)). 
This persistent failure rules out linear weighting as a bottleneck, pointing instead towards intrinsic conflicts within the parameter space.

To further explore the underlying mechanism of this conflict, we visualize parameter interference by analyzing the distribution of cosine similarities between the parameter row vectors of the two models. 
For the compatible \textit{Book-Movie} pair (Figure~\ref{fig:para_vis}(a)), the distribution approximates a zero-mean Gaussian, indicating that most parameter vectors are orthogonal, thereby allowing different domain knowledge to coexist with minimal interference. 
In stark contrast, the divergent \textit{Sport-Toy} pair (Figure~\ref{fig:para_vis}(b)) exhibits a right-skewed distribution with a mean of 0.265. 
This significant deviation from orthogonality reveals that the parameter vectors are geometrically entangled, causing severe collision when linearly superimposed and destroying the specialized preference structures of the target domain.
Therefore, it is imperative to address the following challenge:

\vspace{5pt}

\noindent \textbf{Challenge 1: How to align geometrically incompatible sharp minima to mitigate parameter interference and negative transfer?}

\subsubsection{Rapid Saturation in Multi-domain Fusion}

Practical CDSR scenarios often necessitate integrating diverse auxiliary domains to bolster the target domain. 
However, our empirical investigation uncovers a significant scalability bottleneck.
As illustrated in Figure~\ref{fig:intro}(c), simply increasing the number of source domains yields diminishing returns, with performance gains rapidly hitting a distinct plateau. 
This saturation suggests an inherent upper bound in existing model merging paradigms, where incorporating more diverse data fails to translate into sustained improvements.

To investigate the root cause of this saturation, we visualize the evolution of parameter distributions as more domains are integrated. 
Specifically, we plot the density of the LoRA parameter updates ($\Delta W$) for merged models with $N=1, 3, 7$ source domains. 
As shown in Figure~\ref{fig:intro}(d), the results reveal a distinct trend toward statistical homogenization. 
For the single-domain model ($N=1$), the parameter distribution is relatively flat with heavy tails, indicating a rich presence of high-magnitude weights that encode domain-specific knowledge. 
However, as $N$ increases to 7, the distribution progressively converges to a sharp, narrow Gaussian form concentrated around zero.
This implies that linear aggregation functions as a mean filter, effectively smoothing out distributional outliers—the salient parameters essential for complex reasoning. 
The erosion of these salient signals restricts the model's expressive capacity, highlighting that the performance ceiling is imposed by the linear fusion mechanism itself. 
Therefore, it is necessary to address the following challenge:

\vspace{5pt}

\noindent \textbf{Challenge 2: How to counteract statistical homogenization and reactivate salient preference signals to surmount the performance ceiling in multi-domain fusion?}

\subsection{Sharpness-aware Geometric Alignment}
\label{sec:sga}
To address Challenge 1, we propose Sharpness-aware Geometric Alignment (SGA), which minimizes the discrepancy between the merged model and individual domain experts to mitigate parameter interference. We define this objective as ensuring that aggregated parameters remain within the low-loss regions of each constituent model. 

Specifically, for a domain $D_k$, the ideal adapter $\theta_k$ should not only minimize the empirical risk on its local dataset $\mathcal{D}_k$ but also maintain stability under the parameter shifts $\Delta \theta$ induced by merging. Formally, the optimization objective for SGA is defined as:
\begin{equation}
\min_{\theta_k} \underbrace{\left( \mathcal{L}(\Phi, \theta_k + \Delta \theta; \mathcal{D}_k) - \mathcal{L}(\Phi, \theta_k; \mathcal{D}_k) \right)}_{\text{Merging Interference Resistance}} + \underbrace{\mathcal{L}(\Phi, \theta_k; \mathcal{D}_k)}_{\text{Domain-Specific Accuracy}},
\label{merging_obj}
\end{equation}
where $\theta_k = \{A_k, B_k\}$ represents the LoRA parameters. Since $\Delta \theta$ is unpredictable during independent fine-tuning, we treat it as a worst-case stochastic perturbation $\boldsymbol{\epsilon}$ within a radius $\rho$. As proven in Appendix ~\ref{app:sga_derivation}, this formulation allows us to reformulate Eq. (\ref{merging_obj}) into a sharpness-aware min-max objective:
\begin{equation}
\min_{\theta_k} \max_{|\boldsymbol{\epsilon}|_2 \leq \rho} \mathcal{L} \left( \Phi, \theta_k + \boldsymbol{\epsilon}; \mathcal{D}_k \right).
\label{sharp_obj}
\end{equation}
To handle the scale variance across Transformer layers, we employ an adaptive perturbation $\hat{\epsilon}$ to approximate the inner maximization:
\begin{equation}
\hat{\epsilon} = \rho \frac{\theta_k^2 \nabla_{\theta_k} \mathcal{L}(\theta_k; \mathcal{D}_k)}{| \nabla_{\theta_k} \mathcal{L}(\theta_k; \mathcal{D}_k) |}.
\end{equation}
By updating the adapters using gradients at the perturbed state $\theta_k + \hat{\epsilon}$, SGA effectively guides the optimization toward a \textbf{flat minimum}. 
This geometric alignment ensures that diverse domain experts converge to connected low-loss basins, preventing the destruction of shared user preference structures during fusion.
Consequently, SGA facilitates the seamless transfer of sequential behavioral patterns across heterogeneous domains, effectively resolving the parameter interference (Challenge 1) that typically hinders robust CDSR.

\paragraph{Theoretical Analysis}

To elucidate the efficacy of SGA, we analyze how geometric flatness mitigates negative transfer by defining the merging interference error ($\delta$):
\begin{equation}
    \delta = \mathcal{L}(\Phi, \lambda \theta_A + (1-\lambda)\theta_B; \mathcal{D}) - \left[ \lambda \mathcal{L}(\Phi, \theta_A; \mathcal{D}) + (1-\lambda)\mathcal{L}(\Phi, \theta_B; \mathcal{D}) \right].
\end{equation}
A positive $\delta$ indicates that the merged parameters reside in a high-loss barrier, signifying severe interference. We establish an upper bound for this error (proof in Appendix ~\ref{app:interference_bound}):
\begin{theorem}[Bound on Merging Interference]
Assuming the loss function $\mathcal{L}$ is twice differentiable, the interference error $\delta$ for merging domain-specific LoRA adapters is strictly bounded by:
\begin{equation}
    |\delta| \leq \frac{1}{2}\lambda(1-\lambda) \underbrace{\left( \sigma(\theta_A) + \sigma(\theta_B) \right)}_{\text{Sharpness}} \cdot \underbrace{\|\theta_A - \theta_B\|^2}_{\text{Domain Divergence}} + \mathcal{O}(\epsilon),
\end{equation}
where $\sigma(\theta) = \lambda_{max}(\nabla^2_{\theta} \mathcal{L})$ represents the spectral norm of the Hessian, quantifying local sharpness.
\label{theorem_1}
\end{theorem}

The theorem reveals that interference in CDSR is driven by the product of domain divergence  $\|\theta_A - \theta_B\|^2$ and Loss Sharpness $\sigma(\theta)$. While divergence is an inherent trait of data heterogeneity, sharpness is a controllable geometric property. By explicitly minimizing $\sigma(\theta)$ during adaptation, SGA ensures the merging path remains within a low-loss basin, theoretically guaranteeing robust cross-domain knowledge fusion.

\subsection{Preference Salience Activation}
\label{sec:psa}

To address Challenge 2, we propose Preference Salience Activation (PSA), which reconstructs the parameter distribution from a homogenized Gaussian form into a heavy-tailed distribution. By increasing the probability density of distributional outliers, this transformation effectively preserves and amplifies the salient weights encoding domain-specific expertise that are otherwise eroded by the linear ensemble average.

Specifically, let $\theta_{merge} = \sum_{k=1}^{K} \lambda_k \theta_k$ be the aggregated LoRA parameters from the SGA phase. To break the over-smoothed convergence center, we implement a stochastic disentanglement procedure by introducing an independent Gaussian noise $G \sim \mathcal{N}(0, \sigma_g^2 I)$:
\begin{equation}
    \tilde{\theta} = \theta_{merge} - G.
\end{equation}
This operation shifts the parameters away from the high-density mean, expanding the geometric capacity required for the subsequent non-linear activation of preference signals.

To reactivate the sparse, high-order neurons suppressed by linear averaging, we apply an element-wise non-linear projection $T(\cdot)$ to $\tilde{\theta}$, inducing a heavy-tailed distribution that preserves salient parameters distant from the mean:
\begin{equation}
    \theta_{PSA} = T(\tilde{\theta}) = \text{sign}(\tilde{\theta}) \cdot |\tilde{\theta}|^\gamma \cdot \left( 1 + \alpha e^{-\beta |\tilde{\theta}|} \right),
\end{equation}
where $0 < \gamma < 1$ regulates tail heaviness, and $\alpha, \beta > 0$ control smoothness. The proof that this transformation induces a heavy-tailed distribution is provided in Appendix ~\ref{app:heavy_tail_proof}. By amplifying large-magnitude outliers while suppressing near-zero noise, PSA restores the saliency of domain-specific features. This mechanism ensures high-fidelity knowledge fusion as domains scale, effectively surmounting the performance ceiling identified in Challenge 2.

\paragraph{Theoretical Analysis}
To theoretically justify how PSA mitigates performance saturation, we analyze the relationship between the parameter distribution and the model's functional capacity. 

\begin{theorem}[PSA Expands Preference Coverage]
Let the model function space coverage be defined as 
\begin{equation}
\mathcal{C}(\mathcal{F}) = \int_{\mathcal{W}} |\det(J_{\Phi}(\mathbf{w}))| p_{\mathbf{w}}(\mathbf{w}) d\mathbf{w},
\end{equation}
where $J_{\Phi}(\mathbf{w})$ is the Jacobian matrix. If the initial parameter distribution $p_{\mathbf{w}}$ is Gaussian, the heavy-tailed distribution $p_{\mathbf{w}''}$ induced by PSA ensures that the coverage of the transformed model $\mathcal{C}_2$ is strictly greater than the original model $\mathcal{C}_1$, i.e., $\mathcal{C}_2 > \mathcal{C}_1$. Proof in Appendix ~\ref{app:coverage_proof}.
\label{theorem_2}
\end{theorem}

The Jacobian determinant $|\det(J_{\Phi}(\mathbf{w}))|$ reflects the model's functional sensitivity to parameter variations. In deep recommendation architectures, parameters in the tail region $\mathcal{W}_T$ (outliers) typically evoke more diverse functional forms and stronger activations, representing complex, long-tail user interests. While linear merging triggers Gaussianization that concentrates probability mass in the low-sensitivity center, PSA redistributes density toward the tail region where the Jacobian determinant is larger. This redistribution maximizes the coverage integral $\mathcal{C}(\mathcal{F})$, theoretically guaranteeing that the merged model recovers the capacity to capture intricate cross-domain preferences and surmounts the performance ceiling of multi-domain fusion.

%% file: body/5-exp.tex
\section{Experiments}
To evaluate the effectiveness of \method{}, we conduct experiments to answer the following Research Questions (RQ):
\begin{itemize}[leftmargin=*]\setlength{\itemsep}{-\itemsep} 
    \item \textbf{RQ1}: How does \method{} perform against CDSR baselines? 
    \item \textbf{RQ2}: Can \method{} mitigate performance saturation as source domains increase?
    \item \textbf{RQ3}: How do the \textbf{SGA} and \textbf{PSA} modules contribute to model performance?
    \item \textbf{RQ4}: Does \method{} effectively enhance single-domain sequential recommenders? 
    \item \textbf{RQ5}: Is \method{} robust to varying ratios of overlapping users? 
    \item \textbf{RQ6}: How do key hyperparameters affect \method{}'s sensitivity?
\end{itemize}

\subsection{Experimental Setup}

\subsubsection{Datasets.}
We conduct experiments on the Amazon Review 2023 dataset, selecting seven distinct domains (Sport, Clothing, Movie, Book, Food, Kitchen, and Toy). 
We construct three primary dual-domain pairs for the main performance comparison: Book$\leftrightarrow$Movie, Kitchen$\leftrightarrow$Food, and Sport$\leftrightarrow$Toy, each simulating a realistic hybrid user scenario. 

We apply a 5-core filtering strategy to ensure data quality, represent each item by its textual ``title'' attribute, and truncate interaction sequences to a length of 6--20 items. 
Following the leave-one-out protocol~\citep{geng2022recommendation, lin2024bridging}, we split the sequences into training, validation, and testing sets at an 8:1:1 ratio, where the last item of each validation/test sequence serves as the prediction target and the preceding items form the historical context.
For the scalability analysis (RQ2), we extract 10,000 interaction sequences per domain. 
For the overlap robustness analysis (RQ5), we follow~\cite{liu2025bridge} to generate controlled datasets with overlap ratios of 80\%, 60\%, 40\%, and 20\%. 
During evaluation, each target item is ranked against a candidate pool of 30 items, comprising the ground truth, 10 hard negatives selected by global co-occurrence frequency, and 19 randomly sampled items. 
To adapt the sequential interaction data for LLM-based inputs, we encapsulate each user's historical context and the candidate pool into a structured conversational prompt, as illustrated in Figure~\ref{fig:prompt_template}.
Detailed dataset statistics are summarized in Table~\ref{tab:dataset}.

\begin{figure}[h]
  \centering
  \includegraphics[width=\linewidth]{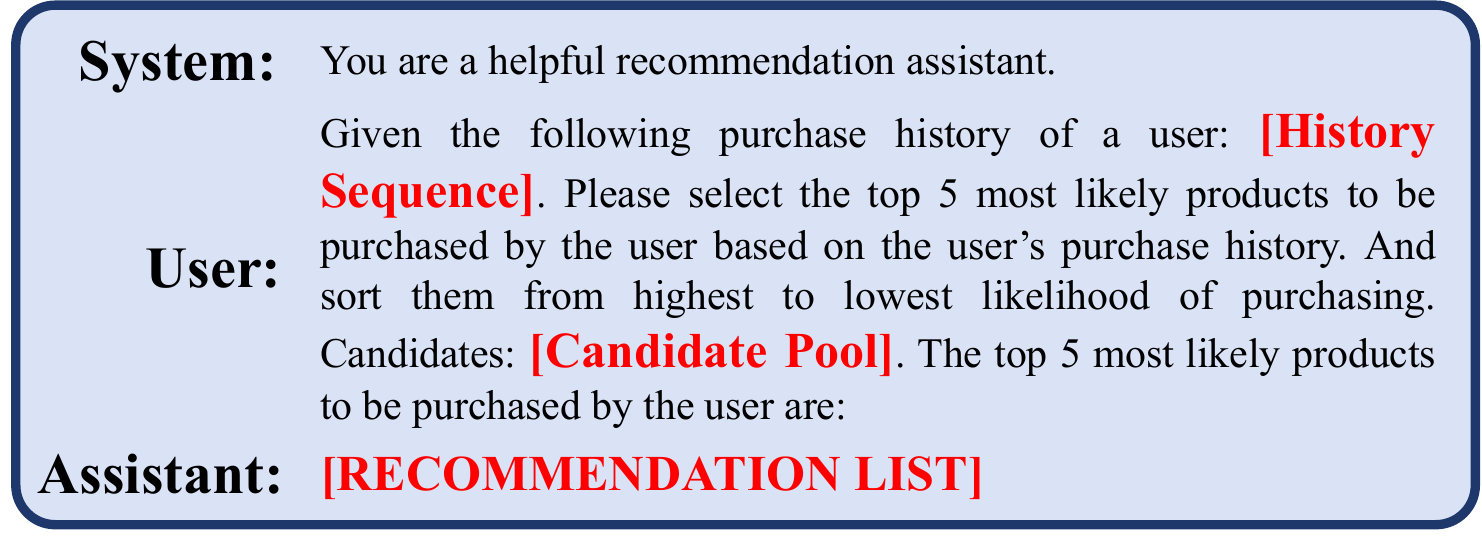}
  \caption{The prompt template for our \method{} framework, illustrating the integration of historical sequences and candidate pools within a conversational LLM structure.}
  \label{fig:prompt_template}
\end{figure}

\begin{table}[t]
\centering
\caption{Dataset statistics.}
\vspace{-0.5em}
\label{tab:dataset}
\resizebox{\linewidth}{!}{
    \begin{tabular}{ccccccc}
        \toprule
        \toprule
        Domain & \# Users & \# Items & \# Interactions & Sparsity & \# Overlap Users & \# Overlap Ratio \\ \midrule
        Book    & 31,271 & 23,133 & 343,514 & 99.95\% & 13,182 & 42.15\% \\
        Movie   & 31,544 & 25,861 & 471,520 & 99.94\% & 13,182 & 41.79\% \\ \midrule
        Kitchen & 22,127 & 9,758  & 149,019 & 99.93\% & 10,015 & 45.26\% \\
        Food    & 22,150 & 11,009 & 190,274 & 99.92\% & 10,015 & 45.21\% \\ \midrule
        Sport   & 10,177 & 6,605  & 115,388 & 99.82\% & 2,171  & 21.33\% \\
        Toy     & 12,104 & 8,940  & 128,046 & 99.88\% & 2,171  & 17.94\% \\
        \bottomrule
        \bottomrule
    \end{tabular}
}
\vspace{-1em}
\end{table}

\subsubsection{Baselines.}
To evaluate the effectiveness of our proposed framework, we compare it against four categories of state-of-the-art methods, including Single-Domain Sequential Recommendation (SDSR), Traditional Cross-Domain Sequential Recommendation (CDSR), LLM-enhanced CDSR, and Model Merging approaches.

\begin{itemize}[leftmargin=*]\setlength{\itemsep}{-\itemsep}
\item \textbf{Single-Domain Sequential Recommendation (SDSR).} (1) \textbf{GRU4Rec}~\cite{hidasi2015session} employs Gated Recurrent Units (GRU) to model sequential interactions within a session. (2) \textbf{SASRec}~\cite{kang2018self} utilizes self-attention mechanisms to capture long-term dependencies in behavior sequences.

\item \textbf{Traditional CDSR.} (3) \textbf{$\pi$-Net}~\cite{ma2019pi} introduces a shared account filter unit and a gating mechanism to transfer knowledge for overlapping users. (4) \textbf{C2DSR}~\cite{cao2022contrastive} employs graph neural networks to jointly learn intra-domain and inter-domain item relationships. (5) \textbf{SyNCRec}~\cite{park2024pacer} mitigates negative transfer via a mixture-of-experts framework equipped with a gradient stop mechanism.

\item \textbf{LLM-based CDSR.} (6) \textbf{URLLM}~\cite{shen2024exploring} integrates LLM-generated semantic representations with traditional collaborative models. (7) \textbf{LLM4CDSR}~\cite{liu2025bridge} aligns LLM-derived knowledge with collaborative signals to enhance cross-domain accuracy. (8) \textbf{LLMCDSR}~\cite{xin2025llmcdsr} leverages LLMs to model universal user preferences to bridge the domain gap.

\item \textbf{Model Merging Methods.} (9) \textbf{Data-Merge} is an instruction-tuning baseline that fine-tunes a single LLM using a unified SFT dataset aggregated from multi-domain interactions. (10) \textbf{X-Cross}~\cite{hadad2025x} integrates fine-tuned domain-specific LoRA modules to facilitate knowledge transfer. (11) \textbf{WeaveRec}~\cite{hou2025weaverec} utilizes LLMs to synthesize cross-domain recommendations through parallel reasoning chains and linear model merging.
\end{itemize}

\subsubsection{Evaluation Metrics}
Following previous works~\cite{cao2022contrastive,  kang2018self}, we adopt commonly used Top-$k$ metrics, specifically Hit Rate (HR@$k$), Normalized Discounted Cumulative Gain (NDCG@$k$), and Mean Reciprocal Rank (MRR), with $k \in \{3, 5\}$ to evaluate recommendation performance. 
For the main experiments, we repeat them five times and report the average results~\footnote{Our code is available at \href{https://github.com/muyiahhh/SharpRec}{https://github.com/muyiahhh/SharpRec}.}. 

\subsubsection{Implementation Details.}

All experiments are implemented using PyTorch and conducted on a high-performance computational cluster equipped with 8 NVIDIA A100 (80GB) and 8 NVIDIA RTX 4090 (24GB) GPUs. 
We utilize the pre-trained Llama-2-7b as the backbone model for all LLM-based methods. 
To ensure fair comparison, the configurations of all baselines are strictly aligned with the optimal settings reported in their respective original papers.

Regarding training configuration and hyperparameters, we employ a global batch size of 8 and a fixed learning rate of $2 \times 10^{-4}$. 
All models are fine-tuned for 2 epochs to ensure convergence while mitigating overfitting. For \method{}, the balancing coefficients of the PSA module are set to $\alpha = 0.1$ and $\beta = 10$. For the core optimization modules, we configure the sharpness-aware perturbation radius at $\rho = 0.01$, the salience activation factor at $\gamma = 0.98$, and the noise injection parameter at $\sigma_g = 0.0001$. 
Sensitivity analyses on $\rho$, $\gamma$, and $\sigma_g$ are further conducted to verify the stability of these configurations, with results reported in Section~\ref{sec:sensitivity}.

\begin{table*}[h]
\centering
\caption[Recommendation performance]{Performance comparison on Cross-Domain Sequential Recommendation tasks.  Best results are typeset in bold with a colored background, e.g., \protect\colorbox[HTML]{95D4EE}{\textbf{80.13$\pm$1.32}}. Runner-up results use a different background color, e.g., \protect\colorbox[HTML]{C8EBF6}{62.34$\pm$0.66}. We run all models 5 times and report the average results and standard deviation. Results are expressed as percentages (\%).}
\vspace{-0.5em}
\label{tab:main_results}
\resizebox{\textwidth}{!}{
\begin{tabular}{ccccccccccc}
    \toprule
    \toprule
    Dataset & \multicolumn{5}{c}{$\bm{Book \longrightarrow Movie}$} & \multicolumn{5}{c}{$\bm{Movie \longrightarrow Book}$} \\ \cmidrule(lr){2-6} \cmidrule(lr){7-11}
    Methods & HR@3 & NDCG@3 & HR@5 & NDCG@5 & MRR & HR@3 & NDCG@3 & HR@5 & NDCG@5 & MRR \\ \midrule
    GRU4Rec & 25.75$\pm$2.12 & 16.52$\pm$0.56 & 38.60$\pm$1.34 & 28.82$\pm$0.87 & 22.24$\pm$1.90 & 27.52$\pm$0.45 & 18.27$\pm$2.32 & 41.16$\pm$1.11 & 27.91$\pm$0.99 & 23.84$\pm$1.54 \\
    SASRec & 28.42$\pm$0.98 & 20.21$\pm$1.65 & 43.57$\pm$2.45 & 26.42$\pm$1.23 & 20.82$\pm$0.67 & 30.71$\pm$2.76 & 21.27$\pm$1.09 & 47.66$\pm$0.34 & 28.23$\pm$1.87 & 21.89$\pm$2.21 \\
    $\pi$-Net & 18.32$\pm$0.25 & 14.64$\pm$1.78 & 26.72$\pm$0.99 & 18.08$\pm$1.56 & 15.27$\pm$2.32 & 20.09$\pm$0.88 & 15.77$\pm$1.90 & 28.00$\pm$0.54 & 19.03$\pm$0.76 & 16.09$\pm$2.21 \\
    C$^2$DSR & 25.60$\pm$2.10 & 18.99$\pm$0.34 & 37.80$\pm$1.67 & 24.00$\pm$0.88 & 24.51$\pm$1.23 & 30.36$\pm$1.54 & 23.02$\pm$0.45 & 41.90$\pm$1.89 & 28.16$\pm$2.11 & 28.17$\pm$0.65 \\
    SyNCRec & 28.41$\pm$0.92 & 24.43$\pm$2.45 & 39.22$\pm$1.10 & 28.84$\pm$0.67 & 31.47$\pm$1.76 & 29.27$\pm$2.01 & 25.12$\pm$1.32 & 39.29$\pm$0.98 & 29.62$\pm$0.43 & 32.35$\pm$1.87 \\
    URLLM & 21.33$\pm$2.11 & 16.99$\pm$0.76 & 30.81$\pm$1.43 & 20.14$\pm$0.55 & 23.29$\pm$1.90 & 22.82$\pm$2.22 & 18.23$\pm$0.33 & 31.36$\pm$1.67 & 22.18$\pm$0.88 & 25.29$\pm$2.76 \\
    LLM4CDSR & \cellcolor[HTML]{C8EBF6}62.34$\pm$0.66 & \cellcolor[HTML]{C8EBF6}53.91$\pm$1.78 & \cellcolor[HTML]{C8EBF6}71.60$\pm$2.65 & \cellcolor[HTML]{C8EBF6}57.72$\pm$1.54 & \cellcolor[HTML]{C8EBF6}53.11$\pm$0.32 & \cellcolor[HTML]{C8EBF6}71.75$\pm$2.43 & \cellcolor[HTML]{C8EBF6}63.51$\pm$1.09 & \cellcolor[HTML]{95D4EE}\textbf{79.61$\pm$0.87} & \cellcolor[HTML]{C8EBF6}66.77$\pm$1.98 & \cellcolor[HTML]{C8EBF6}62.48$\pm$0.44 \\
    LLMCDSR & 35.84$\pm$0.76 & 26.36$\pm$2.12 & 49.18$\pm$0.45 & 29.88$\pm$1.34 & 24.26$\pm$1.11 & 33.62$\pm$2.56 & 24.16$\pm$1.09 & 49.59$\pm$0.32 & 31.56$\pm$1.87 & 25.02$\pm$2.32 \\
    X-Cross & 44.92$\pm$1.25 & 42.15$\pm$0.85 & 50.35$\pm$1.45 & 44.38$\pm$1.22 & 42.45$\pm$1.12 & 40.18$\pm$1.10 & 37.65$\pm$0.98 & 46.02$\pm$1.55 & 40.05$\pm$1.32 & 38.09$\pm$0.95 \\
    WeaveRec & 41.33$\pm$2.89 & 38.57$\pm$1.54 & 46.71$\pm$0.65 & 40.77$\pm$2.76 & 38.84$\pm$1.34 & 36.57$\pm$0.23 & 34.03$\pm$1.87 & 42.38$\pm$2.10 & 36.41$\pm$1.12 & 34.47$\pm$0.88 \\
    Data-Merge & 42.51$\pm$0.55 & 39.48$\pm$2.67 & 47.91$\pm$1.45 & 41.69$\pm$0.32 & 39.66$\pm$2.45 & 38.83$\pm$1.99 & 36.24$\pm$0.76 & 44.15$\pm$1.54 & 38.41$\pm$2.22 & 36.55$\pm$1.23 \\
    \method{} & \cellcolor[HTML]{95D4EE}\textbf{80.13$\pm$1.32} & \cellcolor[HTML]{95D4EE}\textbf{78.98$\pm$0.67} & \cellcolor[HTML]{95D4EE}\textbf{82.46$\pm$2.11} & \cellcolor[HTML]{95D4EE}\textbf{79.93$\pm$1.05} & \cellcolor[HTML]{95D4EE}\textbf{79.11$\pm$0.44} & \cellcolor[HTML]{95D4EE}\textbf{73.24$\pm$1.87} & \cellcolor[HTML]{95D4EE}\textbf{72.15$\pm$0.98} & \cellcolor[HTML]{C8EBF6}75.29$\pm$2.56 & \cellcolor[HTML]{95D4EE}\textbf{72.98$\pm$1.34} & \cellcolor[HTML]{95D4EE}\textbf{72.23$\pm$0.56} \\
    \bottomrule
    Dataset & \multicolumn{5}{c}{$\bm{Kitchen \longrightarrow Food}$} & \multicolumn{5}{c}{$\bm{Food \longrightarrow Kitchen}$} \\
    \cmidrule(lr){2-6} \cmidrule(lr){7-11}
    Methods 
    & HR@3 & NDCG@3 & HR@5 & NDCG@5 & MRR 
    & HR@3 & NDCG@3 & HR@5 & NDCG@5 & MRR \\ \midrule
    GRU4Rec 
    & 19.74$\pm$0.78 & 12.82$\pm$1.45 & 31.06$\pm$0.33 & 20.70$\pm$2.01 & 15.00$\pm$1.10
    & 21.58$\pm$0.82 &  14.21$\pm$1.15 & 33.85$\pm$0.43 &  19.07$\pm$2.12 &  14.22$\pm$1.65 \\
    SASRec 
    & 22.39$\pm$1.12 & 16.04$\pm$0.67 & 35.95$\pm$2.45 & 21.56$\pm$1.34 & 16.89$\pm$0.55
    & 22.21$\pm$1.44 & 16.21$\pm$0.56 & 33.53$\pm$2.89 & 20.85$\pm$0.21 & 16.72$\pm$0.98 \\
    $\pi$-Net 
    & 20.33$\pm$2.56 & 15.37$\pm$1.11 & 33.50$\pm$0.76 & 20.71$\pm$2.01 & 16.58$\pm$0.98
    & 18.17$\pm$1.78 & 13.26$\pm$0.56 & 27.50$\pm$2.89 & 17.06$\pm$1.45 & 13.67$\pm$0.43 \\
    C$^2$DSR 
    & 19.85$\pm$1.23 & 14.68$\pm$0.22 & 30.83$\pm$2.67 & 19.18$\pm$1.45 & 21.84$\pm$0.54
    & 20.47$\pm$0.98 & 15.34$\pm$2.32 & 31.32$\pm$1.54 & 19.78$\pm$0.67 & 21.41$\pm$1.89 \\
    SyNCRec 
    & 23.83$\pm$0.89 & 19.47$\pm$1.78 & 35.72$\pm$0.45 & 23.35$\pm$2.34 & 25.95$\pm$1.21
    & 22.82$\pm$2.11 & 18.54$\pm$1.45 & 34.23$\pm$0.66 & 22.17$\pm$2.89 & 25.44$\pm$1.32 \\
    URLLM 
    & 20.63$\pm$1.21 & 17.87$\pm$0.76 & 31.89$\pm$1.99 & 22.12$\pm$0.55 & 20.36$\pm$2.22
    & 19.78$\pm$1.32 & 16.29$\pm$0.98 & 26.86$\pm$1.87 & 21.76$\pm$0.43 & 19.79$\pm$2.56 \\
    LLM4CDSR 
    & \cellcolor[HTML]{C8EBF6}49.27$\pm$1.78 & 40.46$\pm$0.56 & \cellcolor[HTML]{95D4EE}\textbf{60.44$\pm$2.89} & \cellcolor[HTML]{C8EBF6}45.03$\pm$0.99 & 39.95$\pm$1.21
    & \cellcolor[HTML]{C8EBF6}35.58$\pm$0.65 & 26.83$\pm$1.99 & 45.47$\pm$1.11 & \cellcolor[HTML]{C8EBF6}33.13$\pm$0.45 & 29.08$\pm$2.32 \\
    LLMCDSR 
    & 31.90$\pm$2.12 & 25.66$\pm$1.45 & 41.76$\pm$0.54 & 30.12$\pm$2.34 & 25.07$\pm$1.65
    & 33.74$\pm$1.34 & 26.80$\pm$0.76 & \cellcolor[HTML]{C8EBF6}47.89$\pm$2.54 & 32.04$\pm$1.89 & 26.90$\pm$0.88 \\
    X-Cross 
    & 43.15$\pm$1.34 & 40.22$\pm$1.15 & 48.85$\pm$1.56 & 42.55$\pm$1.89 & 40.50$\pm$1.05 
    & 32.12$\pm$1.12 & \cellcolor[HTML]{C8EBF6}29.45$\pm$0.98 & 38.05$\pm$1.45 & 31.85$\pm$1.23 & 29.15$\pm$1.15 \\
    WeaveRec 
    & 44.47$\pm$2.21 & \cellcolor[HTML]{C8EBF6}41.60$\pm$1.45 & 50.15$\pm$0.54 & 43.95$\pm$1.98 & \cellcolor[HTML]{C8EBF6}41.91$\pm$0.90
    & 30.43$\pm$0.88 & 26.86$\pm$2.56 & 36.67$\pm$1.32 & 29.40$\pm$0.65 & 27.03$\pm$1.89 \\
    Data-Merge 
    & 40.62$\pm$1.78 & 37.86$\pm$2.34 & 45.79$\pm$1.21 & 39.98$\pm$0.76 & 38.07$\pm$1.65
    & 31.61$\pm$1.45 & 28.91$\pm$0.99 & 37.40$\pm$2.87 & 31.28$\pm$1.56 & \cellcolor[HTML]{C8EBF6}29.30$\pm$0.87 \\
    \method{} 
    & \cellcolor[HTML]{95D4EE}\textbf{55.27$\pm$0.89} & \cellcolor[HTML]{95D4EE}\textbf{53.17$\pm$1.76} & \cellcolor[HTML]{C8EBF6}60.19$\pm$0.44 & \cellcolor[HTML]{95D4EE}\textbf{55.19$\pm$2.32} & \cellcolor[HTML]{95D4EE}\textbf{53.56$\pm$1.11}
    & \cellcolor[HTML]{95D4EE}\textbf{44.71$\pm$2.22} & \cellcolor[HTML]{95D4EE}\textbf{42.26$\pm$0.45} & \cellcolor[HTML]{95D4EE}\textbf{50.96$\pm$1.67} & \cellcolor[HTML]{95D4EE}\textbf{44.82$\pm$2.98} & \cellcolor[HTML]{95D4EE}\textbf{42.84$\pm$1.34} \\
    \bottomrule
    Dataset & \multicolumn{5}{c}{$\bm{Toy \longrightarrow Sport}$} & \multicolumn{5}{c}{$\bm{Sport \longrightarrow Toy}$} \\ \cmidrule(lr){2-6} \cmidrule(lr){7-11}
    Methods & HR@3 & NDCG@3 & HR@5 & NDCG@5 & MRR & HR@3 & NDCG@3 & HR@5 & NDCG@5 & MRR \\ \midrule
    GRU4Rec & 28.59$\pm$1.56 & 22.04$\pm$0.34 & 39.12$\pm$2.11 & 27.44$\pm$1.67 & 23.62$\pm$0.55 & 27.77$\pm$1.12 & 19.97$\pm$0.45 & 41.30$\pm$2.31 & 31.56$\pm$1.05 & 29.01$\pm$0.89 \\
    SASRec & 25.79$\pm$0.23 & 18.34$\pm$1.89 & 40.71$\pm$0.78 & 24.45$\pm$2.45 & 19.15$\pm$1.32 & 31.20$\pm$0.98 & 22.97$\pm$1.21 & 45.93$\pm$2.67 & 29.02$\pm$0.43 & 23.50$\pm$2.01 \\
    $\pi$-Net & 21.84$\pm$1.54 & 16.74$\pm$0.65 & 34.10$\pm$2.87 & 21.82$\pm$1.11 & 17.84$\pm$0.99 & 17.70$\pm$2.34 & 13.34$\pm$0.56 & 26.84$\pm$1.78 & 17.09$\pm$0.88 & 13.91$\pm$1.45 \\
    C$^2$DSR & 26.97$\pm$0.87 & 21.41$\pm$2.12 & 37.11$\pm$1.34 & 25.58$\pm$0.21 & 27.35$\pm$1.98 & 32.18$\pm$1.05 & 24.13$\pm$0.32 & 43.57$\pm$2.54 & 28.83$\pm$1.65 & 28.90$\pm$0.76 \\
    SyNCRec & 28.48$\pm$2.45 & 23.66$\pm$1.67 & 39.00$\pm$0.54 & 25.41$\pm$2.32 & 30.68$\pm$1.10 & 26.78$\pm$0.89 & 21.93$\pm$1.45 & 37.48$\pm$0.22 & 26.31$\pm$2.10 & 27.76$\pm$1.32 \\
    URLLM & 22.76$\pm$1.23 & 17.27$\pm$2.54 & 30.28$\pm$1.89 & 20.31$\pm$0.87 & 21.27$\pm$2.65 & 24.11$\pm$1.12 & 20.28$\pm$0.56 & 33.68$\pm$1.90 & 22.87$\pm$0.33 & 24.82$\pm$1.67 \\
    LLM4CDSR & \cellcolor[HTML]{C8EBF6}62.23$\pm$2.34 & 50.83$\pm$0.89 & \cellcolor[HTML]{C8EBF6}73.40$\pm$1.56 &  55.44$\pm$1.12 & 49.46$\pm$2.01 & \cellcolor[HTML]{C8EBF6}60.74$\pm$1.76 & \cellcolor[HTML]{C8EBF6}52.20$\pm$0.55 & \cellcolor[HTML]{C8EBF6}70.54$\pm$2.34 & \cellcolor[HTML]{C8EBF6}56.25$\pm$0.67 & \cellcolor[HTML]{C8EBF6}51.51$\pm$1.45 \\
    LLMCDSR & 40.98$\pm$1.99 & 29.99$\pm$0.55 & 50.88$\pm$1.23 & 34.91$\pm$2.87 & 29.41$\pm$1.45 & 37.20$\pm$0.67 & 28.42$\pm$2.12 & 53.92$\pm$1.54 & 34.70$\pm$0.90 & 27.44$\pm$1.76 \\
    X-Cross & 53.98$\pm$1.32 & 50.85$\pm$1.87 & 61.50$\pm$1.56 & \cellcolor[HTML]{C8EBF6}55.63$\pm$1.22 & 51.20$\pm$1.65 & 49.85$\pm$1.45 & 45.20$\pm$1.12 & 57.50$\pm$1.34 & 48.35$\pm$0.98 & 45.40$\pm$1.55 \\
    WeaveRec & 53.68$\pm$0.54 & \cellcolor[HTML]{C8EBF6}50.99$\pm$2.56 & 59.95$\pm$1.45 & 53.55$\pm$0.32 & \cellcolor[HTML]{C8EBF6}51.47$\pm$1.89 & 48.58$\pm$2.32 & 43.52$\pm$1.21 & 55.26$\pm$0.65 &  46.29$\pm$1.77 & 43.33$\pm$2.65 \\
    Data-Merge & 51.28$\pm$1.76 & 48.56$\pm$0.88 & 57.86$\pm$2.98 & 51.25$\pm$1.65 & 49.11$\pm$0.77 & 48.65$\pm$1.32 & 43.94$\pm$2.54 & 56.39$\pm$1.09 & 47.13$\pm$0.45 & 44.10$\pm$1.90 \\
    \method{} & \cellcolor[HTML]{95D4EE}\textbf{76.14$\pm$2.32} & \cellcolor[HTML]{95D4EE}\textbf{74.26$\pm$0.43} & \cellcolor[HTML]{95D4EE}\textbf{78.78$\pm$1.56} & \cellcolor[HTML]{95D4EE}\textbf{75.34$\pm$2.87} & \cellcolor[HTML]{95D4EE}\textbf{74.21$\pm$1.21} & \cellcolor[HTML]{95D4EE}\textbf{66.34$\pm$0.98} & \cellcolor[HTML]{95D4EE}\textbf{63.88$\pm$1.67} & \cellcolor[HTML]{95D4EE}\textbf{71.66}$\pm$0.34 & \cellcolor[HTML]{95D4EE}\textbf{66.05$\pm$2.11} & \cellcolor[HTML]{95D4EE}\textbf{64.23$\pm$1.05} \\
    \bottomrule
    \bottomrule
\end{tabular}
}
\end{table*}

\subsection{Results and Discussions}
\subsubsection{Overall Performance Comparison (RQ1)}
Table~\ref{tab:main_results} reports the performance of \method{} and competing baselines across three cross-domain datasets. Key observations are summarized below:

\begin{itemize}[leftmargin=*]\setlength{\itemsep}{-\itemsep}
\item \textbf{\method{} consistently yields the best performance across all six cross-domain tasks.}
The significant improvement over all baselines validates that our sharpness-aware alignment and salience recovery mechanisms successfully mitigate interference and enhance knowledge transfer.

\item \textbf{Traditional CDSR models struggle to outperform single-domain baselines in hybrid user scenarios.}
Their reliance on explicit user overlap leads to ineffective knowledge transfer when co-occurrence signals are sparse, as evidenced by $\pi$-Net and $C^2$DSR often lagging behind SASRec.

\item \textbf{LLM-enhanced methods generally outperform traditional baselines by leveraging semantic reasoning to address the hybrid user scenario.} By utilizing open-world knowledge, these methods effectively bridge domain gaps even without explicit overlap. However, URLLM remains an exception, performing poorly due to its dependence on fragile user retrieval mechanisms.

\item \textbf{Current model merging paradigms are limited by statistical homogenization and rapid saturation.}
Linear aggregation in methods like WeaveRec dilutes distinct domain-specific features into a mediocre Gaussian distribution. In contrast, \method{} reactivates these salient signals through non-linear reparameterization, effectively lifting the performance ceiling.

\end{itemize}

\subsubsection{Scalability Analysis (RQ2)}
We evaluate the scalability of \method{} by incrementally increasing the number of source domains. 
As illustrated in Figure~\ref{fig:saturation}, we conclude that \textbf{\method{} effectively mitigates the performance saturation bottleneck inherent in multi-domain fusion.} 
Specifically, while baseline methods (e.g., WeaveRec) rapidly plateau entering the Saturation Zone after integrating 3-4 domains, \method{} maintains a robust upward trajectory, achieving continuous gains. 
This saturation in baselines stems from the \textit{statistical homogenization} caused by linear aggregation, which tends to dilute high-order cross-domain signals into a mediocre Gaussian distribution.
In contrast, \method{} overcomes this limitation by employing the PSA mechanism to reconstruct the merged parameters into a heavy-tailed distribution, thereby reactivating the salient features essential for sustained knowledge transfer and scalability.

\begin{figure}[t]
    \subfigure[\textbf{Food}]{\includegraphics[width=0.475\linewidth]{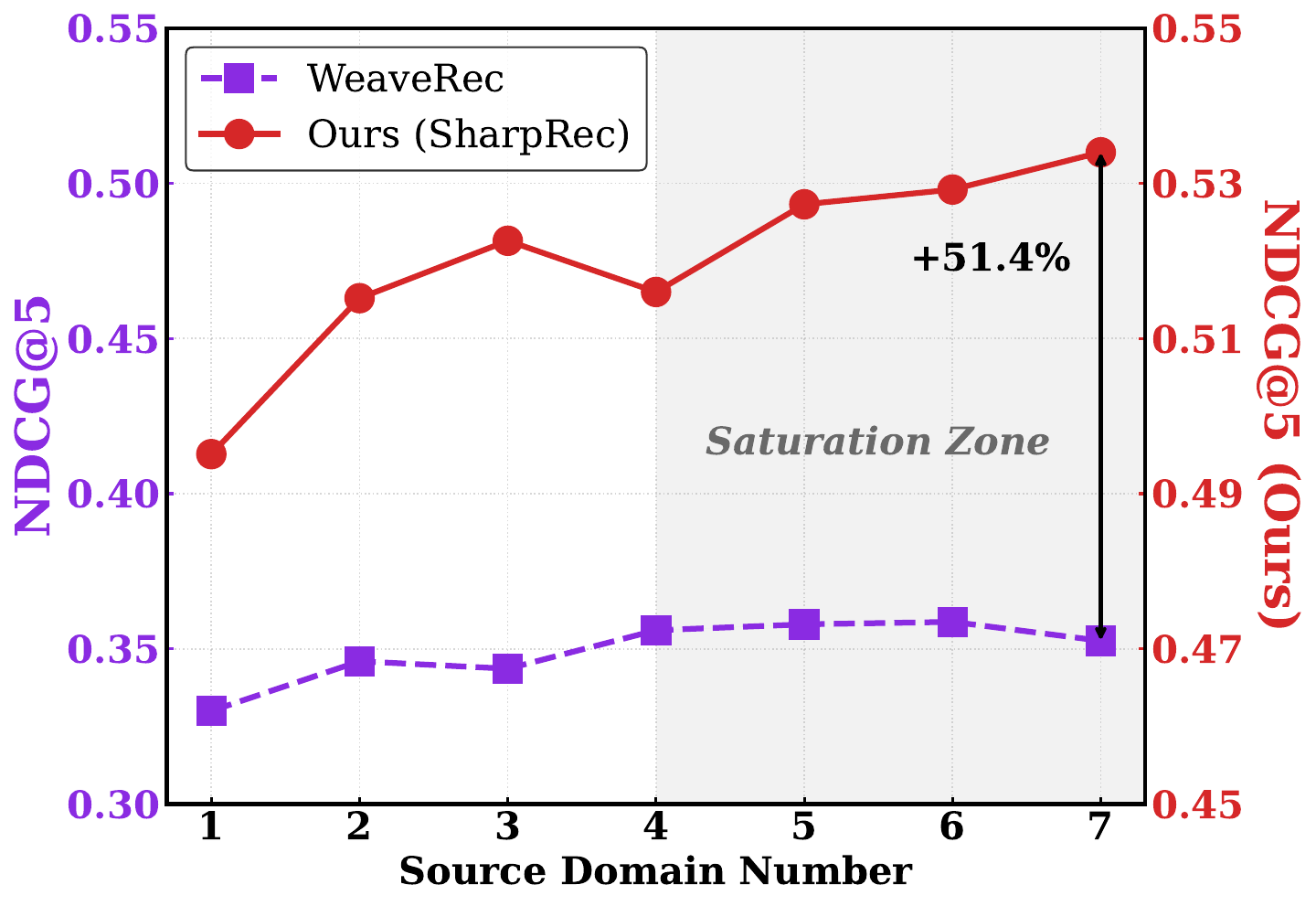}} \label{fig:hp-filmtrust-ns}
    \subfigure[\textbf{Sport}]{\includegraphics[width=0.475\linewidth]{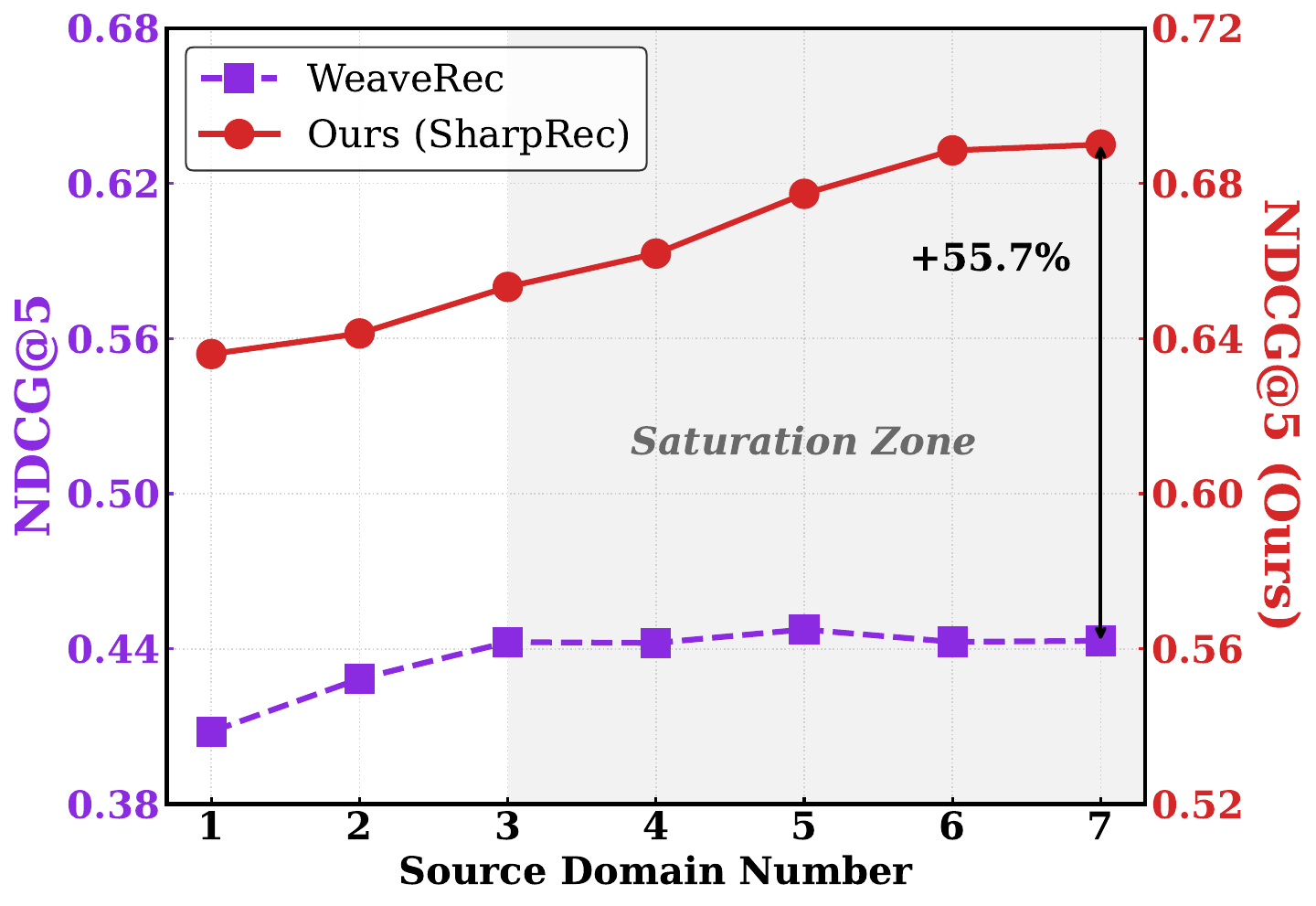}} \label{fig:hp-lastfm-ns}
    \vspace{-1.0em}
    \caption{Performance comparison w.r.t. the number of merged source domains on Food and Sport tasks.}
    \label{fig:saturation}
    \vspace{-1.25em}
\end{figure}

\subsection{Ablation Study}
\subsubsection{Component Analysis (RQ3)} 
To evaluate the individual contributions of our proposed modules, we conduct ablation studies across both dual-domain and multi-domain scenarios. 
As reported in Table~\ref{tab:ablation_gsr_merge}, we compare \method{} with two variants: 
(1) \textbf{w/o SGA}, which removes the sharpness-aware geometric alignment; 
and (2) \textbf{w/o PSA}, which excludes the preference salience activation. 
Our findings are as follows:
\begin{itemize}[leftmargin=*]\setlength{\itemsep}{-\itemsep}
    \item \textbf{SGA serves as the foundation for resolving geometric incompatibility.} 
    The performance degradation observed in the \textit{w/o SGA} variant confirms that direct merging leads to severe parameter interference. 
    By guiding domain-specific models into a unified flat loss landscape, SGA ensures a stable initialization, preventing parameter interference and negative transfer.

    \item \textbf{PSA alleviates statistical homogenization to prevent performance saturation in multi-domain fusion.} 
    While less critical in dual-domain settings, PSA becomes essential as domains increase by preventing merged parameters from converging to a mediocre average that dilutes domain-specific characteristics. 
    It employs non-linear reparameterization to recover heavy-tailed distributions, preserving distinct preference signals required for robust cross-domain transfer.
\end{itemize}

\begin{table}[t]
\centering
\caption{Impact of SGA and PSA on model performance.}
\vspace{-0.5em}
\label{tab:ablation_gsr_merge}
\resizebox{\linewidth}{!}{
\begin{tabular}{lccccc}

\toprule
\toprule
Method & NDCG@3 & HR@3 & NDCG@5 & HR@5 & MRR \\
\midrule
\rowcolor[HTML]{E6E6FA} 
\multicolumn{6}{c}{$\bm{Sport \longrightarrow Toy}$} \\
\method{} & \textbf{0.6388} & \textbf{0.6634} & \textbf{0.6605} & \textbf{0.7166} & \textbf{0.6423} \\
w/o SGA   & 0.4564 & 0.5014 & 0.4817 & 0.5632 & 0.4548 \\
w/o PSA   &0.6110 & 0.6442 & 0.6290 & 0.6889 & 0.6094 \\
\rowcolor[HTML]{E6E6FA}
\midrule
\multicolumn{6}{c}{$\bm{Book, Movie, Kitchen, Toy, Food, Cloth \longrightarrow Sport}$} \\
\method{} & \textbf{0.6757} & \textbf{0.7344} & \textbf{0.6900} & \textbf{0.7638} & \textbf{0.6749} \\
w/o SGA   & 0.4384 & 0.5094 & 0.4673 & 0.5500 & 0.4451 \\
w/o PSA   & 0.6328 & 0.6959 & 0.6485 & 0.7219 & 0.6295 \\
\bottomrule
\bottomrule
\end{tabular}
}
\end{table}

\subsubsection{Cross-domain Gain Analysis (RQ4)} 
To verify the effectiveness of cross-domain knowledge transfer via model merging, we compare \method{} with single-domain baselines. As shown in Table~\ref{tab:ablation_sharp_single_domain}, we conclude that \textbf{\method{} consistently enhances recommendation performance across all target domains compared to single-domain baselines.} 
Results demonstrate substantial performance lift in every CDSR task, validating the fundamental premise of CDSR that leveraging auxiliary domain interactions provides complementary preference signals inaccessible to isolated single-domain models.
This superiority confirms that \method{} effectively mitigates the negative transfer often seen in vanilla merging, employing sharpness-aware alignment and salience recovery to filter interference while distilling synergistic knowledge.

\begin{table}[t]
\centering
\caption{Performance comparison of \method{} and single-domain baselines under different domain transfer settings.}
\vspace{-0.5em}
\label{tab:ablation_sharp_single_domain}
\resizebox{\linewidth}{!}{
\begin{tabular}{llccccc}
\toprule
\toprule
Domain & Method & NDCG@3 & HR@3 & NDCG@5 & HR@5 & MRR \\
\midrule
\rowcolor[HTML]{E6E6FA}
\multicolumn{7}{c}{$\bm{Book \longleftrightarrow Movie}$} \\
\multirow{2}{*}{$\bm{Book}$}
 & Book-only & 0.7128 & 0.7253 & 0.7199 & 0.7425 & 0.7125 \\
 & \method{} & \textbf{0.7215} & \textbf{0.7324} & \textbf{0.7298} & \textbf{0.7529} & \textbf{0.7223} \\
\multirow{2}{*}{$\bm{Movie}$}
 & Movie-only & 0.7827 & 0.7929 & 0.7921 & 0.8158 & 0.7844 \\
 & \method{}  & \textbf{0.7898} & \textbf{0.8013} & \textbf{0.7993} & \textbf{0.8246} & \textbf{0.7911} \\
\midrule
\rowcolor[HTML]{E6E6FA}
\multicolumn{7}{c}{$\bm{Kitchen \longleftrightarrow Food}$} \\
\multirow{2}{*}{$\bm{Kitchen}$}
 & Kitchen-only & 0.3988 & 0.4291 & 0.4192 & 0.4792 & 0.3996 \\
 & \method{}   & \textbf{0.4226} & \textbf{0.4471} & \textbf{0.4482} & \textbf{0.5096} & \textbf{0.4284} \\
\multirow{2}{*}{$\bm{Food}$}
 & Food-only & 0.5171 & 0.5380 & 0.5377 & 0.5887 & 0.5211 \\
 & \method{} & \textbf{0.5317} & \textbf{0.5527} & \textbf{0.5519} & \textbf{0.6019} & \textbf{0.5356} \\
\midrule
\rowcolor[HTML]{E6E6FA}
\multicolumn{7}{c}{$\bm{Sport \longleftrightarrow Toy}$} \\
\multirow{2}{*}{$\bm{Sport}$}
 & Sport-only & 0.7263 & 0.7428 & 0.7350 & 0.7645 & 0.7254 \\
 & \method{} & \textbf{0.7426} & \textbf{0.7614} & \textbf{0.7534} & \textbf{0.7878} & \textbf{0.7421} \\
\multirow{2}{*}{$\bm{Toy}$}
 & Toy-only & 0.6015 & 0.6300 & 0.6222 & 0.6804 & 0.6032 \\
 & \method{} & \textbf{0.6388} & \textbf{0.6634} & \textbf{0.6605} & \textbf{0.7166} & \textbf{0.6423} \\
\bottomrule
\bottomrule
\end{tabular}
}

\end{table}

\subsection{Robustness to Overlap Ratios (RQ5)} 
To evaluate whether \method{} maintains stable performance under varying data constraints, we investigate its robustness by adjusting the ratio of overlapping users from 20\% to 80\%. 
\textbf{The experimental results demonstrate that \method{} consistently achieves superior and stable performance regardless of the user overlap ratio, effectively circumventing the overlap dilemma.}
As illustrated in Figure~\ref{fig:overlap_exp}, while traditional methods like $C^2$DSR exhibit notable performance degradation at low overlap ratios (e.g., 20\%) due to their heavy reliance on dense co-occurrence signals, \method{} maintains a substantial lead across all settings.
This robustness confirms that by leveraging semantic alignment and resolving geometric incompatibility in the parameter space, our framework ensures effective knowledge transfer even in scenarios with minimal user overlap, decoupling performance from cross-domain data density.

\begin{figure}[t] 
    \centering
    \includegraphics[width=0.8\linewidth]{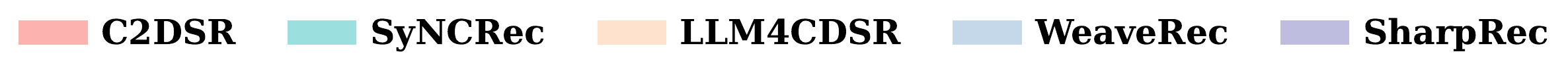} \\

    \subfigure[Book Domain]{\includegraphics[width=0.48\linewidth]{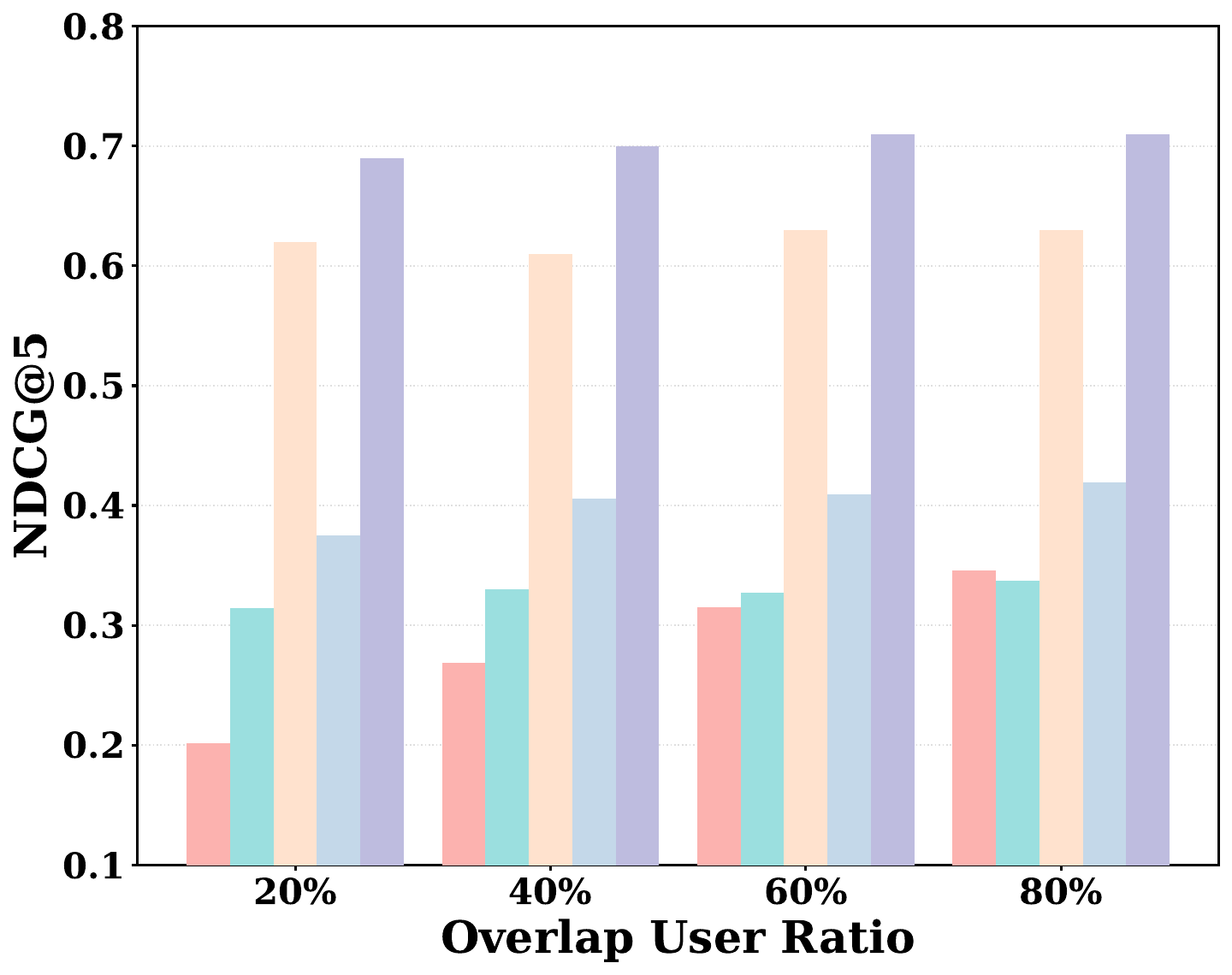}}
    \hfill 
    \subfigure[Movie Domain]{\includegraphics[width=0.48\linewidth]{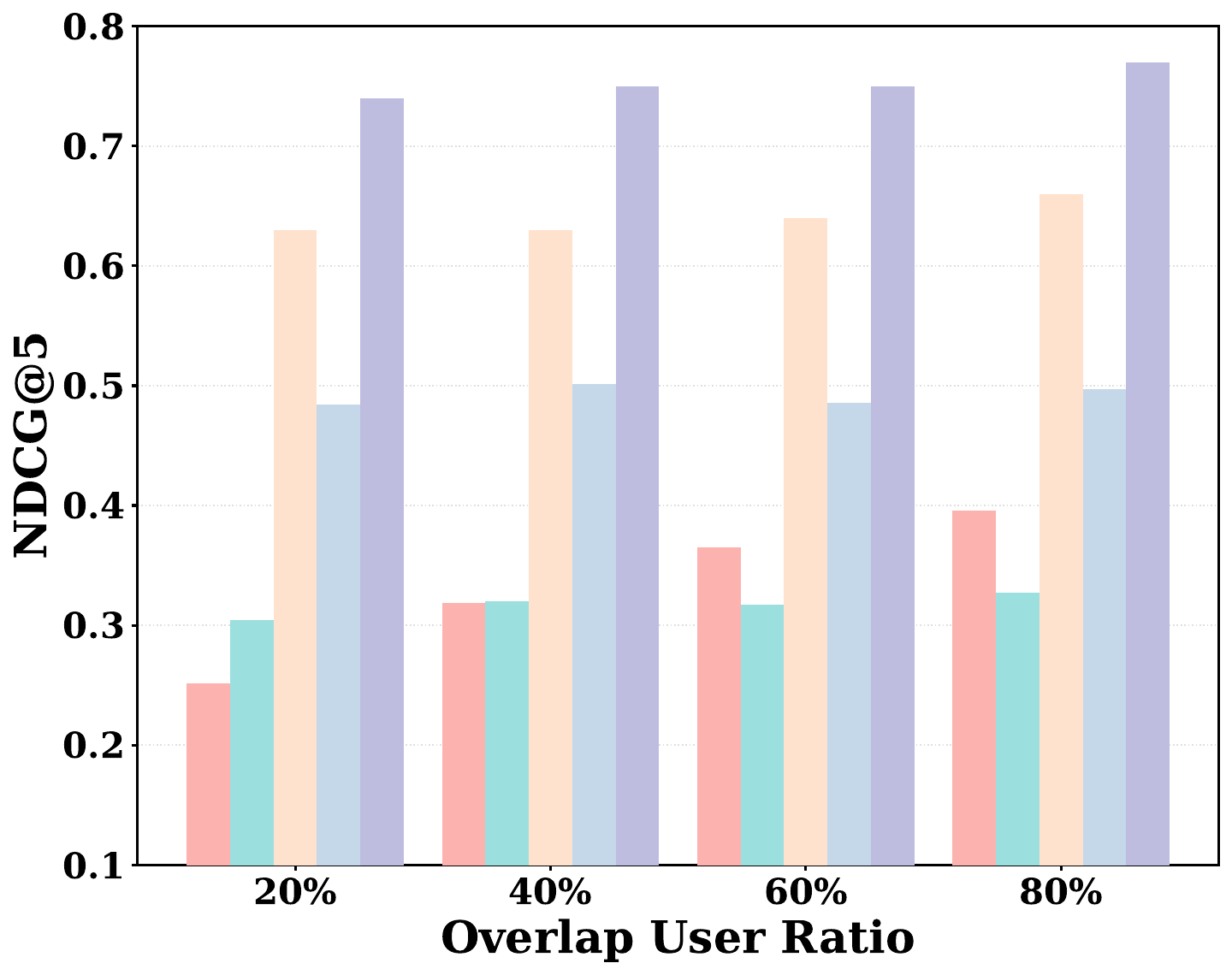}}
    \caption{Impact of user overlap ratios (ranging from 20\% to 80\%) on recommendation performance (NDCG@5).}
    \label{fig:overlap_exp}
    \vspace{-1.0em}
\end{figure}

\subsection{Hyper-parameter Sensitivity (RQ6)}
\label{sec:sensitivity}
To investigate the robustness of \method{}, we conduct sensitivity analyses on three key hyperparameters: the perturbation radius $\rho$, the salience activation factor $\gamma$, and the noise variance $\sigma_g$. 
As illustrated in Figure ~\ref{fig:hyper-para}, our observations are as follows:
%
%

\begin{itemize}[leftmargin=*]\setlength{\itemsep}{-\itemsep}
    \item \textbf{Impact of Perturbation Radius $\rho$.} 
    \method{} exhibits robust performance within the range of $[0.001, 0.05]$. 
    Values below this range fail to propel the model out of domain-specific sharp minima, hindering effective alignment in the shared preference space. Conversely, $\rho > 0.05$ introduces excessive perturbation that disrupts the semantic coherence required for cross-domain transfer.

    \item \textbf{Effect of Salience Factor $\gamma$.} 
    Performance peaks as $\gamma$ approaches $[0.94, 0.98]$, confirming that a moderate degree of reparameterization is optimal for preserving the preference manifold. 
    Deviating from this range (e.g., $\gamma < 0.90$) leads to over-amplification, which distorts the structure of transferable user interests and compromises the fidelity of the merged model.

    \item \textbf{Sensitivity of Noise Variance $\sigma_g$.} 
    Small-scale Gaussian noise acts as an effective regularizer, facilitating the learning of domain-invariant representations during salience activation. 
    However, once $\sigma_g$ surpasses a critical threshold (e.g., $\sigma_g > 0.01$), the excessive variance destabilizes the feature alignment between domains, resulting in a consistent drop in accuracy.
\end{itemize}

\begin{figure}[t]
    \centering
    \vspace{-0.25em}
    \includegraphics[width=0.35\linewidth]{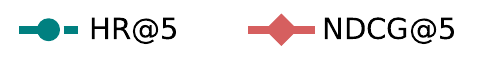} \\
    \vspace{-0.5em}

    \subfigure[$\rho$, Book Domain]
    {\includegraphics[width=0.235\textwidth]{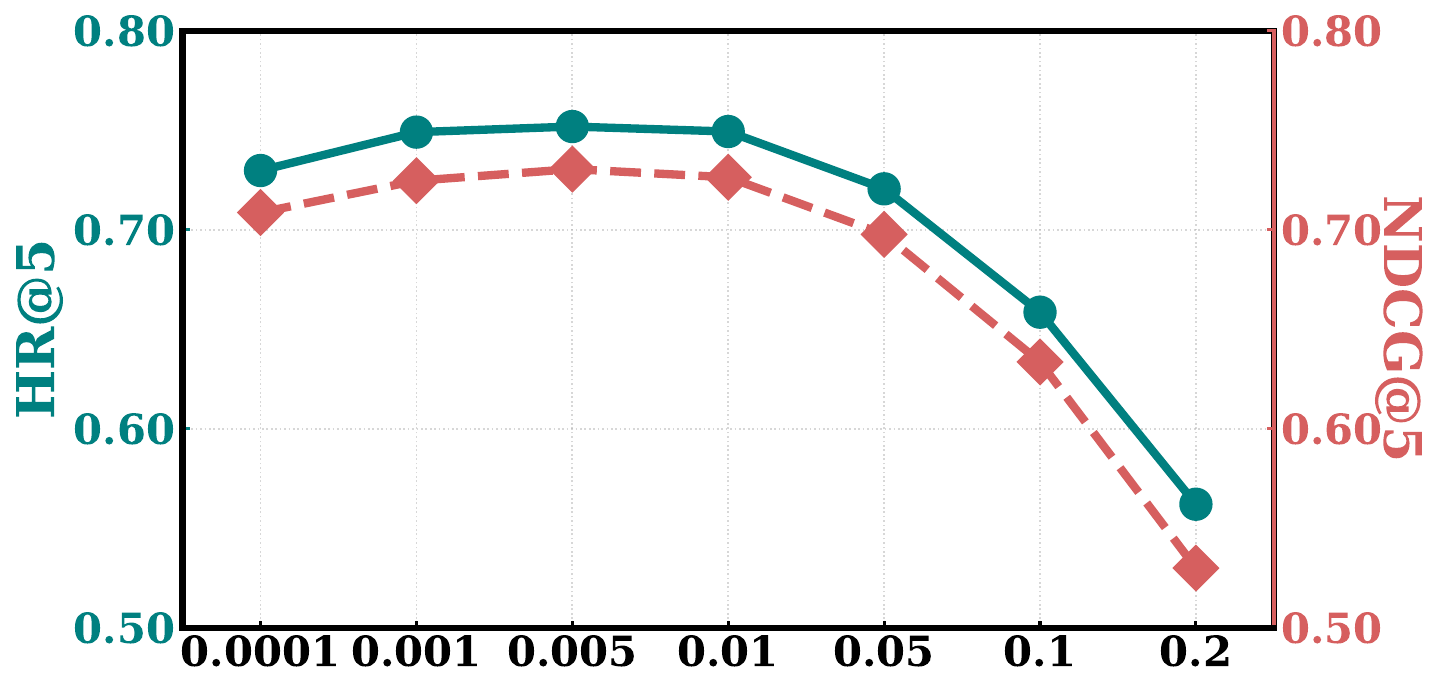}}
    \hfill
    \subfigure[$\rho$, Movie Domain]{\includegraphics[width=0.235\textwidth]{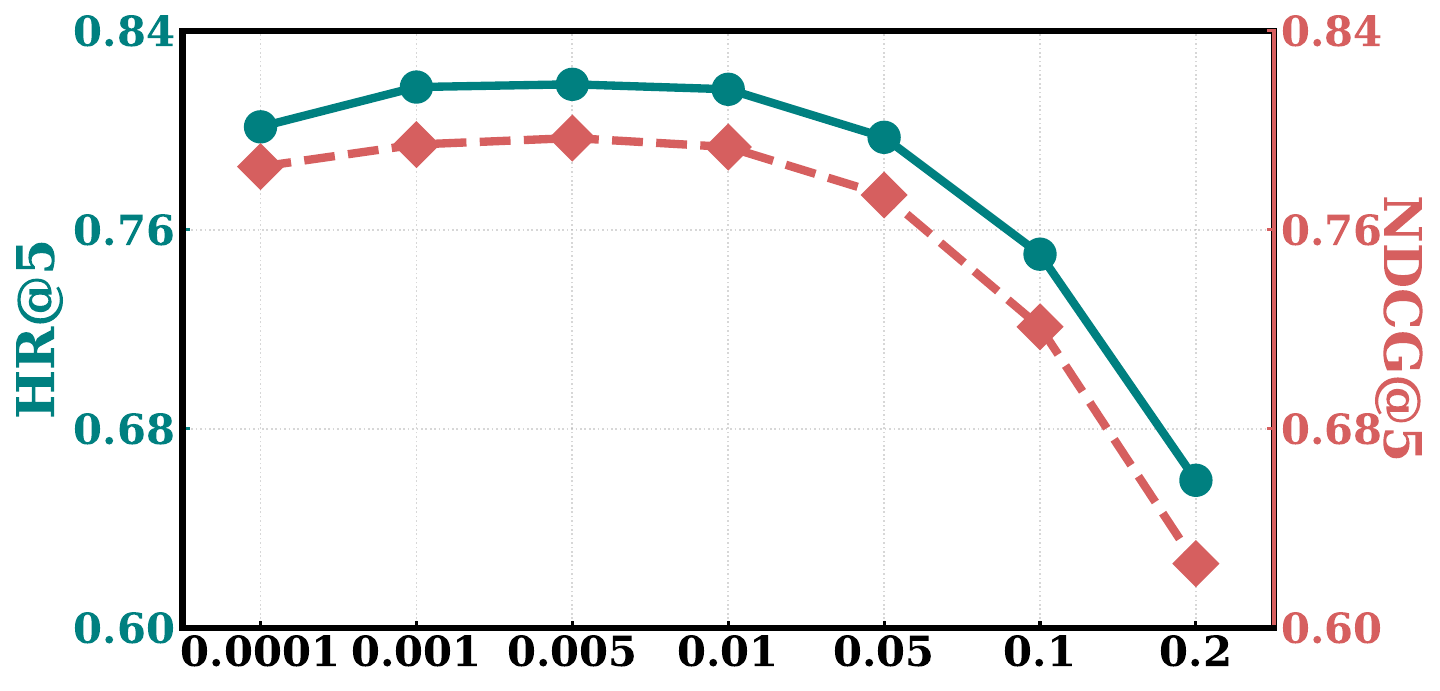}}
    \hfill
    \subfigure[$\gamma$, Book Domain]
    {\includegraphics[width=0.235\textwidth]{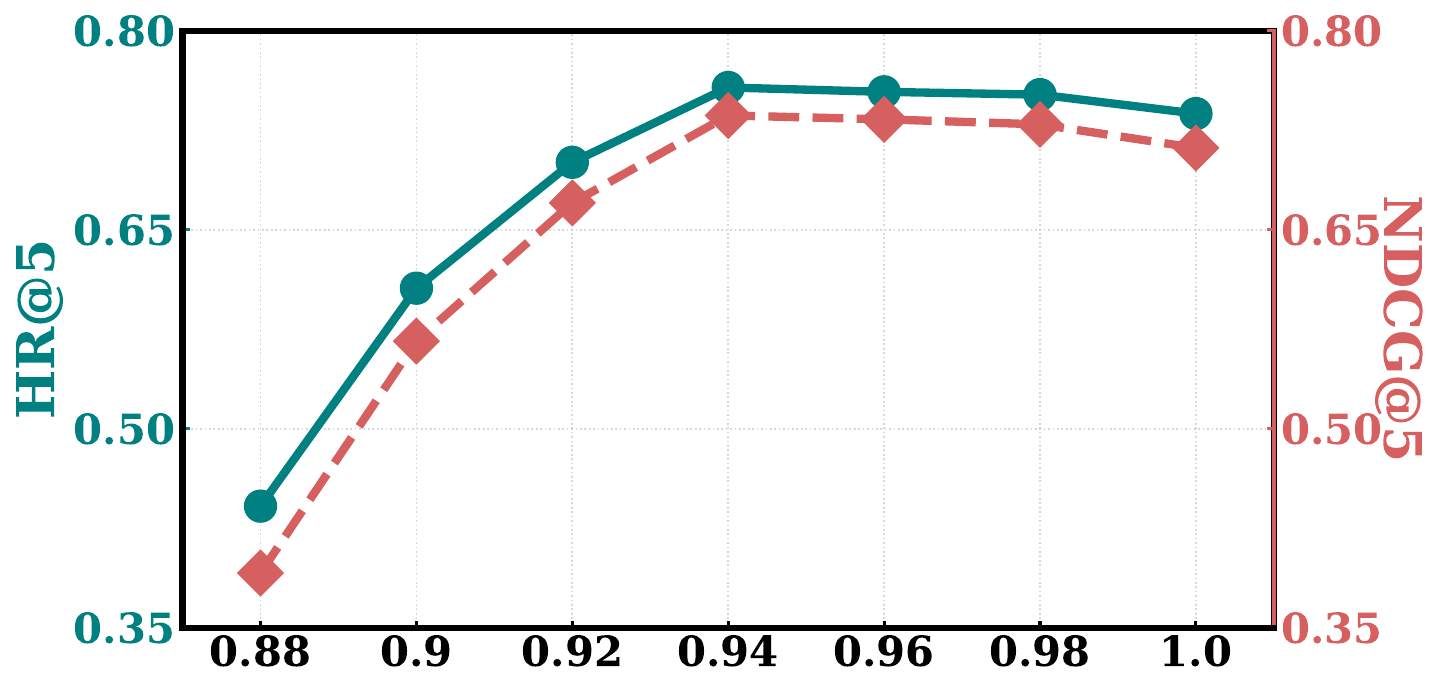}}
    \hfill
    \subfigure[$\gamma$, Movie Domain]{\includegraphics[width=0.235\textwidth]{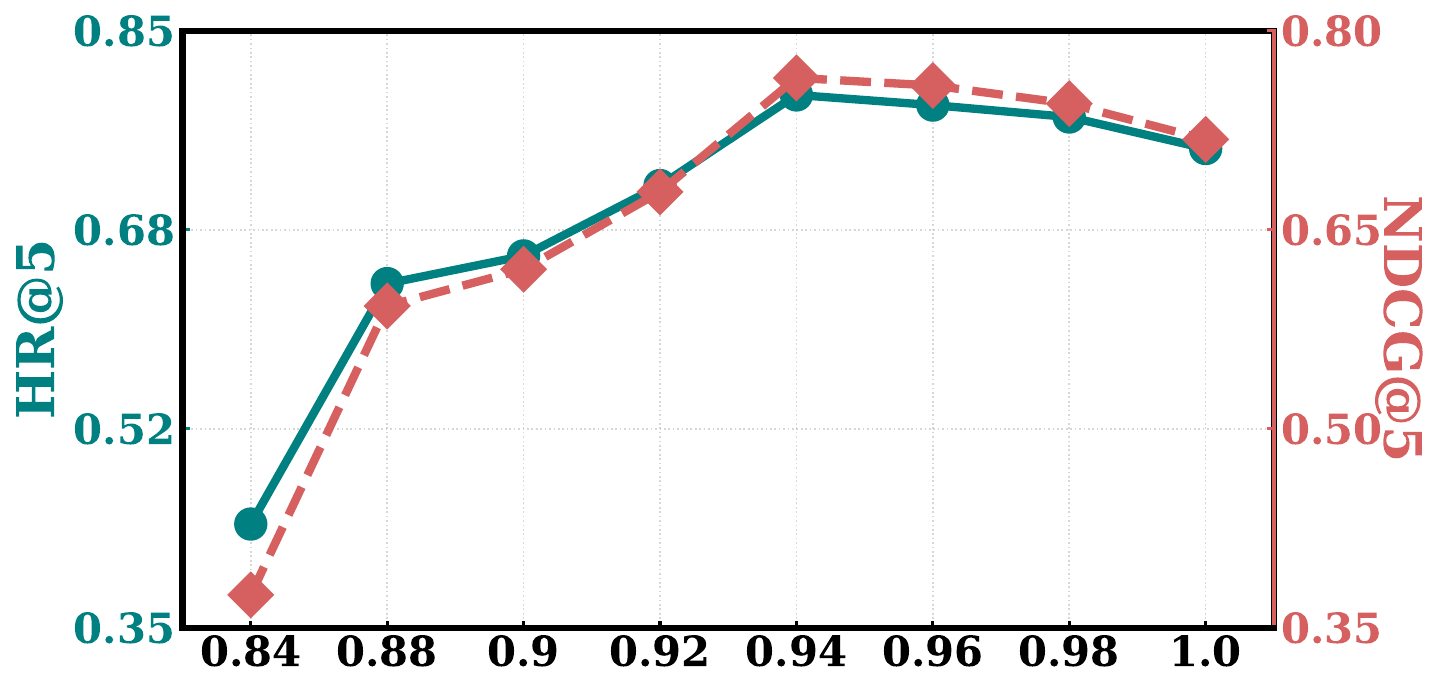}}
    \hfill
    \subfigure[$\sigma_g$, Book Domain]
    {\includegraphics[width=0.235\textwidth]{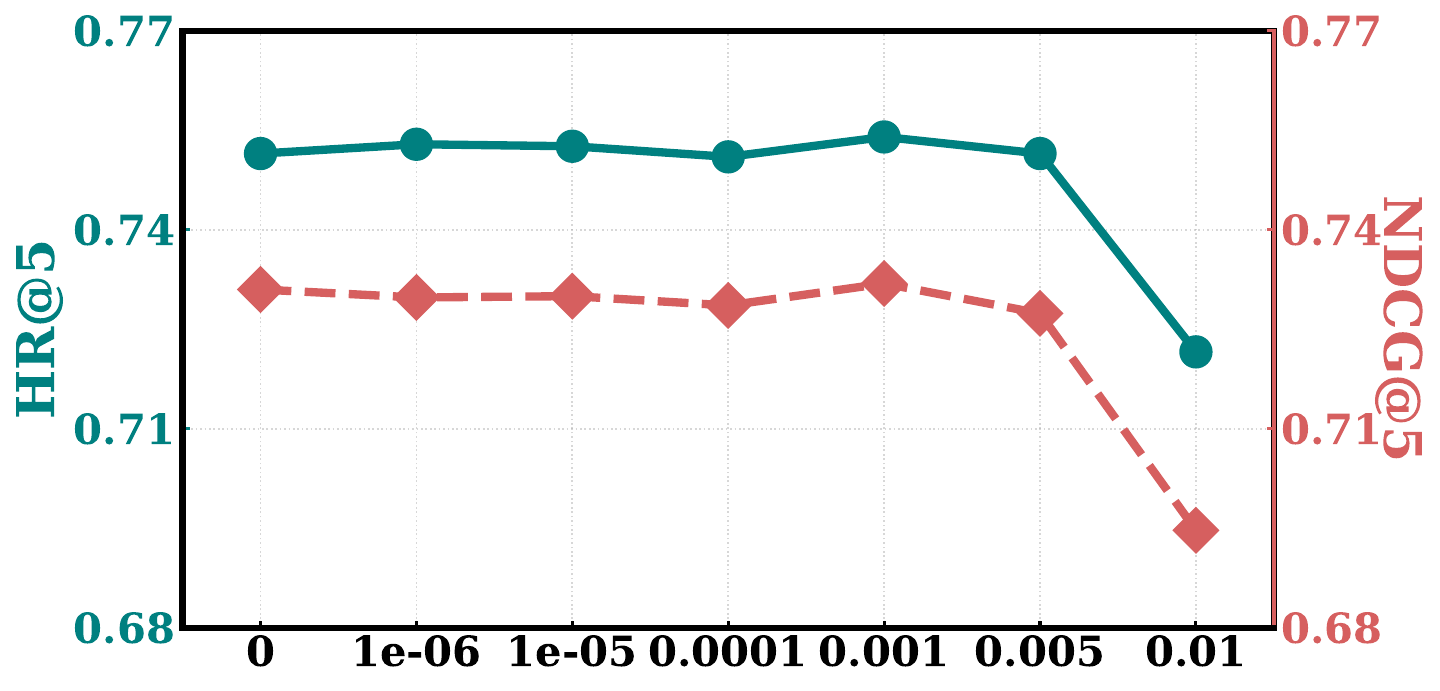}}
    \hfill
    \subfigure[$\sigma_g$, Movie Domain]{\includegraphics[width=0.235\textwidth]{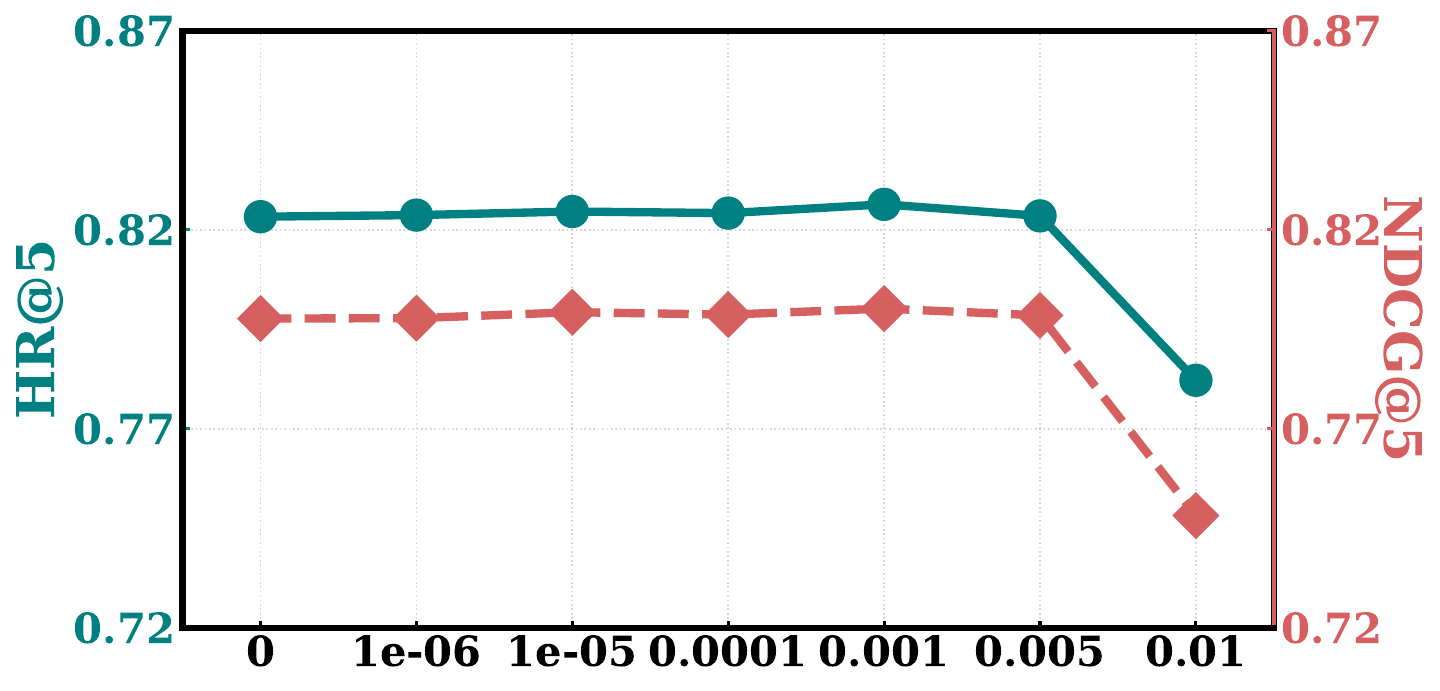}}
    \hfill
    \caption{Impact of hyperparameters $\rho$, $\gamma$, and $\sigma_g$ on recommendation performance (HR@5 and NDCG@5).}
    \label{fig:hyper-para}
\end{figure}

%% file: body/6-con.tex
\section{Conclusion}
In this work, we explore the potential of model merging for CDSR. 
Our empirical studies identify two critical bottlenecks in existing paradigms, i.e., cross-domain knowledge conflict arising from geometric incompatibility and performance saturation caused by statistical homogenization. 
To address these bottlenecks, we proposed \method{}, a novel framework designed to lift the performance upper bound of CDSR. 
Specifically, \method{} leverages Sharpness-aware Geometric Alignment to resolve parameter interference via connected flat minima, and Preference Salience Activation to reactivate salient signals by reconstructing heavy-tailed distributions.
Extensive experiments verify that \method{} effectively mitigates cross-domain conflicts and alleviates performance saturation, significantly enhancing target domain recommendation performance while enabling continuous scalability as source domains increase.
Our work provides a robust and scalable solution for LLM-based CDSR, paving the way for more efficient utilization of open-world knowledge in recommender systems.

%% file: body/app.tex
\section{Derivation of the SGA Objective}
\label{app:sga_derivation}

\textit{Adapted from the proof of Eq.~(6) in~\citep{lee2025mitigating} (Appendix~B),
reformulated for the cross-domain LoRA adapter merging setting in CDSR.}

\noindent\textbf{Simplification of the Optimization Objective.} 
We start with the proposed SGA objective function for a domain-specific adapter $\theta_k$, which aims to simultaneously minimize the empirical loss and the merging interference:
\begin{equation}
    \mathcal{J}_{\text{SGA}}(\theta_k) = \underbrace{\left( \mathcal{L}(\Phi, \theta_k + \Delta \theta; \mathcal{D}_k) - \mathcal{L}(\Phi, \theta_k; \mathcal{D}_k) \right)}_{\text{Merging Interference Resistance}} + \underbrace{\mathcal{L}(\Phi, \theta_k; \mathcal{D}_k)}_{\text{Domain-Specific Accuracy}}.
\end{equation}
By performing algebraic simplification, the term $\mathcal{L}(\Phi, \theta_k; \mathcal{D}_k)$ cancels out, reducing the objective to minimizing the loss at the shifted parameter state:
\begin{equation}
\begin{aligned}
    \mathcal{J}_{\text{SGA}}(\theta_k) &= \mathcal{L}(\Phi, \theta_k + \Delta \theta; \mathcal{D}_k) - \mathcal{L}(\Phi, \theta_k; \mathcal{D}_k) + \mathcal{L}(\Phi, \theta_k; \mathcal{D}_k) \\
    &= \mathcal{L}(\Phi, \theta_k + \Delta \theta; \mathcal{D}_k).
\end{aligned}
\label{eq:simplified_obj}
\end{equation}
This reveals that minimizing interference is equivalent to optimizing the model's performance under the parameter shift $\Delta \theta$ induced by the merging process.

\noindent\textbf{Modeling the Parameter Shift $\Delta \theta$.} 
In linear model merging, the aggregated parameters are $\theta_{merge} = \sum_{j=1}^{K} \lambda_j \theta_j$. We decompose this sum relative to the current domain $\theta_k$ to identify the shift $\Delta \theta$:
\begin{equation}
\begin{split}
    \Delta \theta & = \theta_{merge} - \theta_k \\
                  & = \left( \lambda_k \theta_k + \sum_{j \neq k} \lambda_j \theta_j \right) - \theta_k \\
                  & = (\lambda_k - 1)\theta_k + \sum_{j \neq k} \lambda_j \theta_j
    \label{eq:shift_expansion}.
\end{split}
\end{equation}
Here, $\Delta \theta$ comprises a scaling residual of the current adapter and a linear superposition of peer domain adapters.

\noindent\textbf{Equivalence to Min-Max Optimization.} 
During the independent fine-tuning of domain $k$, the parameters of other domains $\{\theta_j\}_{j \neq k}$ are unknown and dynamic. Consequently, the term $\sum_{j \neq k} \lambda_j \theta_j$ in Eq. (\ref{eq:shift_expansion}) acts as an unpredictable stochastic variable.

To guarantee robustness against \textit{any} potential configuration of other domains, we treat $\Delta \theta$ as a perturbation vector $\boldsymbol{\epsilon}$ bounded within a local geometric region. We reformulate the optimization as a Min-Max problem, seeking the parameter $\theta_k$ that minimizes the loss under the worst-case shift within a Euclidean ball of radius $\rho$:
\begin{equation}
    \min_{\theta_k} \mathcal{L}(\Phi, \theta_k + \Delta \theta; \mathcal{D}_k) \implies \min_{\theta_k} \max_{\|\boldsymbol{\epsilon}\|_2 \leq \rho} \mathcal{L}(\Phi, \theta_k + \boldsymbol{\epsilon}; \mathcal{D}_k).
\end{equation}
This demonstrates that minimizing the merging interference (Eq. \ref{merging_obj}) is mathematically equivalent to seeking a flat minimum via Sharpness-Aware Minimization.

\section{Proof of Theorem~\ref{theorem_1}}
\label{app:interference_bound}

\textit{Following the Taylor-expansion bounding technique of~\citep{lee2025mitigating}
(Theorem~1, Appendix~D), reformulated for the CDSR setting.}

\noindent\textbf{Setup and Definitions. }
We aim to bound the interference error $\delta$:
\begin{equation}
    \delta = \mathcal{L}(\lambda \theta_A + (1-\lambda)\theta_B) - \left[ \lambda \mathcal{L}(\theta_A) + (1-\lambda)\mathcal{L}(\theta_B) \right].
\end{equation}
Let $\theta_{merge} = \lambda \theta_A + (1-\lambda)\theta_B$. We assume the domain experts $\theta_A$ and $\theta_B$ have converged to local minima, implying negligible first-order gradients: $\nabla \mathcal{L}(\theta_A) \approx 0$ and $\nabla \mathcal{L}(\theta_B) \approx 0$.

\noindent\textbf{Quadratic Expansion. }
We perform a second-order Taylor expansion of $\mathcal{L}(\theta_{merge})$ around $\theta_A$. Noting that the displacement vector is $\theta_{merge} - \theta_A = (1-\lambda)(\theta_B - \theta_A)$, we have:
\begin{equation}
    \mathcal{L}(\theta_{merge}) \approx \mathcal{L}(\theta_A) + \frac{1}{2} (1-\lambda)^2 (\theta_A - \theta_B)^\top \mathbf{H}_A (\theta_A - \theta_B),
\end{equation}
where $\mathbf{H}_A = \nabla^2 \mathcal{L}(\theta_A)$. Rearranging terms yields the loss difference relative to $\theta_A$:
\begin{equation}
    \mathcal{L}(\theta_{merge}) - \mathcal{L}(\theta_A) \approx \frac{1}{2} (1-\lambda)^2 (\theta_A - \theta_B)^\top \mathbf{H}_A (\theta_A - \theta_B).
\end{equation}
By symmetry, the expansion around $\theta_B$ with displacement $\theta_{merge} - \theta_B = -\lambda(\theta_A - \theta_B)$ yields:
\begin{equation}
    \mathcal{L}(\theta_{merge}) - \mathcal{L}(\theta_B) \approx \frac{1}{2} \lambda^2 (\theta_A - \theta_B)^\top \mathbf{H}_B (\theta_A - \theta_B).
\end{equation}

\noindent\textbf{Deriving the Upper Bound. }
Substituting these quadratic forms back into the expression for $\delta$:
\begin{equation}
\begin{aligned}
    \delta \approx &\lambda \left[ \frac{1}{2} (1-\lambda)^2 (\theta_A - \theta_B)^\top \mathbf{H}_A (\theta_A - \theta_B) \right] \\
           &+ (1-\lambda) \left[ \frac{1}{2} \lambda^2 (\theta_A - \theta_B)^\top \mathbf{H}_B (\theta_A - \theta_B) \right].
\end{aligned}
\end{equation}
Factorizing common terms $\frac{1}{2}\lambda(1-\lambda)$:
\begin{equation}
\begin{split}
    \delta \approx \frac{1}{2}\lambda(1-\lambda) \Bigl[ &(1-\lambda) (\theta_A - \theta_B)^\top \mathbf{H}_A (\theta_A - \theta_B) \\
    &+ \lambda (\theta_A - \theta_B)^\top \mathbf{H}_B (\theta_A - \theta_B) \Bigr].
\end{split}
\end{equation}
To bound this scalar value, we utilize the spectral norm of the Hessian, denoted as $\sigma(\theta) = \lambda_{max}(\mathbf{H})$. By the definition of the Rayleigh quotient, for any vector $v$, $v^\top \mathbf{H} v \leq \lambda_{max}(\mathbf{H}) \|v\|^2$. Applying this property:
\begin{equation}
    |\delta| \leq \frac{1}{2}\lambda(1-\lambda) \|\theta_A - \theta_B\|^2 \left[ (1-\lambda)\sigma(\theta_A) + \lambda\sigma(\theta_B) \right] + \mathcal{O}(\epsilon).
\end{equation}
Since $\lambda \in [0,1]$, the convex combination of sharpness terms is strictly bounded by their sum: $(1-\lambda)\sigma(\theta_A) + \lambda\sigma(\theta_B) \leq \sigma(\theta_A) + \sigma(\theta_B)$.
This yields the final bound presented in Theorem ~\ref{theorem_1}:
\begin{equation}
    |\delta| \leq \frac{1}{2}\lambda(1-\lambda) \underbrace{\left( \sigma(\theta_A) + \sigma(\theta_B) \right)}_{\text{Sharpness}} \cdot \underbrace{\|\theta_A - \theta_B\|^2}_{\text{Domain Divergence}} + \mathcal{O}(\epsilon).
\end{equation}
This concludes the proof.

\section{Proof of Heavy-Tailed Distribution Induction}
\label{app:heavy_tail_proof}

\textit{Adapted from the proof of Theorem~5 in~\citep{wang2025more}(Appendix~C.5), reframed for the PSA module in \method{}.}

\noindent\textbf{Gaussian Property of Disentangled Parameters. }
Let the aggregated parameters $\theta_{merge}$ follow a multivariate Gaussian distribution $\theta_{merge} \sim \mathcal{N}(\mu, \Sigma)$. We introduce an independent Gaussian noise variable $G \sim \mathcal{N}(0, \sigma_g^2 I)$. The disentangled representation is defined as $\tilde{\theta} = \theta_{merge} - G$.
Since the linear combination of independent Gaussian vectors remains Gaussian, we have:
\begin{equation}
    \tilde{\theta} \sim \mathcal{N}(\mu, \Sigma + \sigma_g^2 I).
\end{equation}
Assuming a simplified diagonal covariance for element-wise analysis, each component $x \in \tilde{\theta}$ follows $\mathcal{N}(0, \sigma^2)$, with the probability density function:
\begin{equation}
    p_{\tilde{\theta}}(x) = \frac{1}{\sqrt{2\pi}\sigma} \exp\left( -\frac{x^2}{2\sigma^2} \right).
\end{equation}

\noindent\textbf{Non-linear Transformation and Asymptotic Behavior. }
We analyze the distribution of the transformed parameter $\theta_{PSA}$ after applying the element-wise projection $T(\cdot)$. Let $y = T(x)$ denote a component of $\theta_{PSA}$. The transformation is given by:
\begin{equation}
    y = T(x) = \text{sign}(x) \cdot |x|^\gamma \cdot \left( 1 + \alpha e^{-\beta |x|} \right).
\end{equation}
In the tail region where $|x|$ is large, the term $e^{-\beta |x|}$ decays rapidly to zero. Thus, the transformation implies the following asymptotic relationship and its inverse:
\begin{equation}
    y \approx \text{sign}(x) \cdot |x|^\gamma \implies x \approx \text{sign}(y) \cdot |y|^{1/\gamma}.
\end{equation}

\noindent\textbf{Formulation of the Resulting Distribution. }
We derive probability density function of $y$, denoted as $p_{\theta_{PSA}}(y)$, using the change of variables formula:
\begin{equation}
    p_{\theta_{PSA}}(y) = p_{\tilde{\theta}}(x) \cdot \left| \frac{dx}{dy} \right|.
\end{equation}
Using the asymptotic inverse $|x| \approx |y|^{1/\gamma}$, the Jacobian of the transformation is:
\begin{equation}
    \left| \frac{dx}{dy} \right| \approx \frac{d}{dy} \left( |y|^{1/\gamma} \right) = \frac{1}{\gamma} |y|^{\frac{1}{\gamma} - 1}.
\end{equation}
Substituting the Gaussian PDF $p_{\tilde{\theta}}(x)$ and the Jacobian term into the change of variables equation:
\begin{equation}
    p_{\theta_{PSA}}(y) \approx \frac{1}{\sqrt{2\pi}\sigma} \exp\left( -\frac{(|y|^{1/\gamma})^2}{2\sigma^2} \right) \cdot \frac{1}{\gamma} |y|^{\frac{1}{\gamma} - 1}.
\end{equation}
Simplifying the expression yields the asymptotic distribution for $\theta_{PSA}$:
\begin{equation}
    p_{\theta_{PSA}}(y) \propto |y|^{\frac{1}{\gamma} - 1} \exp\left( -\frac{|y|^{2/\gamma}}{2\sigma^2} \right).
\end{equation}
For $0 < \gamma < 1$, this density function belongs to the family of Generalized Error Distributions (GED). The parameter $\gamma$ controls the shape parameter $p = 2/\gamma$, allowing the PSA module to explicitly reshape the statistical properties of the aggregated model parameters beyond the original Gaussian limitations.

\section{Proof of Theorem ~\ref{theorem_2}}
\label{app:coverage_proof}

\textit{Adapted from the proof of Theorem~6 in~\citep{wang2025more}(Appendix~C.6), reframed for the PSA coverage analysis in \method{}.}

\noindent\textbf{Parameter Space Partition and Sensitivity. }
We define the functional coverage $\mathcal{C}$ as the volume of the function space accessible by the parameter distribution, weighted by the Jacobian determinant:
\begin{equation}
    \mathcal{C} = \int_{\mathcal{W}} |\det(J_{\Phi}(\mathbf{w}))| p(\mathbf{w}) d\mathbf{w},
\end{equation}
where $J_{\Phi}(\mathbf{w})$ is the Jacobian matrix of the mapping $\Phi: \mathcal{W} \rightarrow \mathcal{F}$.

We consider the parameter space $\mathcal{W}$ as the union of two disjoint regions: a high-density central region $\mathcal{W}_C$ (near zero) and a tail region $\mathcal{W}_T$ (outliers). In deep recommendation models, parameter sensitivity is non-uniform. Parameters in $\mathcal{W}_T$ typically represent strong activation signals or specialized preferences, inducing larger functional variations than those in the inactive central region. Mathematically, this property implies:
\begin{equation}
    |\det(J_{\Phi}(\mathbf{w}))|_{\mathbf{w} \in \mathcal{W}_T} > |\det(J_{\Phi}(\mathbf{w}))|_{\mathbf{w} \in \mathcal{W}_C}.
\end{equation}

\noindent\textbf{Proof of Coverage Expansion. }
Let $p_{\text{Gauss}}(\mathbf{w})$ denote the original distribution from linear merging, and $p_{\text{PSA}}(\mathbf{w})$ denote the transformed heavy-tailed distribution.
As derived in Appendix ~\ref{app:heavy_tail_proof}, the PSA transformation functions as a mass transport mechanism: it shifts probability mass from the low-sensitivity center $\mathcal{W}_C$ to the high-sensitivity tail $\mathcal{W}_T$.

We compare the coverage integrals $\mathcal{C}_{\text{PSA}}$ and $\mathcal{C}_{\text{Gauss}}$ by analyzing their difference:
\begin{equation}
    \Delta \mathcal{C} = \int_{\mathcal{W}} |\det(J_{\Phi}(\mathbf{w}))| \left( p_{\text{PSA}}(\mathbf{w}) - p_{\text{Gauss}}(\mathbf{w}) \right) d\mathbf{w}.
\end{equation}
Splitting the integral over the two regions:
\begin{equation}
\begin{split}
    \Delta \mathcal{C} = & \int_{\mathcal{W}_T} |\det(J_{\Phi})| \underbrace{(p_{\text{PSA}} - p_{\text{Gauss}})}_{>0} d\mathbf{w} \\
    & + \int_{\mathcal{W}_C} |\det(J_{\Phi})| \underbrace{(p_{\text{PSA}} - p_{\text{Gauss}})}_{<0} d\mathbf{w}.
\end{split}
\end{equation}
Since the total probability mass is conserved ($\int \Delta p = 0$), the mass added to the tail equals the mass removed from the center. However, because the weighting factor $|\det(J_{\Phi})|$ is strictly larger in the tail region ($\mathcal{W}_T$) than in the central region ($\mathcal{W}_C$), the positive contribution from the first integral strictly dominates the negative contribution from the second.

Therefore, we have $\Delta \mathcal{C} > 0$, which implies:
\begin{equation}
    \int_{\mathcal{W}} |\det(J_{\Phi}(\mathbf{w}))| p_{\text{PSA}}(\mathbf{w}) d\mathbf{w} > \int_{\mathcal{W}} |\det(J_{\Phi}(\mathbf{w}))| p_{\text{Gauss}}(\mathbf{w}) d\mathbf{w}.
\end{equation}
This confirms that $\mathcal{C}_{\text{PSA}} > \mathcal{C}_{\text{Gauss}}$, proving that the heavy-tailed distribution expands the effective model coverage.

%% file: sample-base.bib
@String{Computing = "Computing" }

@String{Computer = "{IEEE} Computer" }

@String{Springer = "Springer-Verlag" }

@ArtifactSoftware{R,
    title = {R: A Language and Environment for Statistical Computing},
    author = {{R Core Team}},
    organization = {R Foundation for Statistical Computing},
    address = {Vienna, Austria},
    year = {2019},
    url = {https://www.R-project.org/},
}

@article{hidasi2015session,
  title={Session-based recommendations with recurrent neural networks},
  author={Hidasi, Bal{\'a}zs and Karatzoglou, Alexandros and Baltrunas, Linas and Tikk, Domonkos},
  journal={arXiv preprint arXiv:1511.06939},
  year={2015}
}

@inproceedings{kang2018self,
  title={Self-attentive sequential recommendation},
  author={Kang, Wang-Cheng and McAuley, Julian},
  booktitle={2018 IEEE international conference on data mining (ICDM)},
  pages={197--206},
  year={2018},
  organization={IEEE}
}

@inproceedings{ma2019pi,
  title={$\pi$-net: A parallel information-sharing network for shared-account cross-domain sequential recommendations},
  author={Ma, Muyang and Ren, Pengjie and Lin, Yujie and Chen, Zhumin and Ma, Jun and Rijke, Maarten de},
  booktitle={Proceedings of the 42nd international ACM SIGIR conference on research and development in information retrieval},
  pages={685--694},
  year={2019}
}

@inproceedings{cao2022contrastive,
  title={Contrastive cross-domain sequential recommendation},
  author={Cao, Jiangxia and Cong, Xin and Sheng, Jiawei and Liu, Tingwen and Wang, Bin},
  booktitle={Proceedings of the 31st ACM International Conference on Information \& Knowledge Management},
  pages={138--147},
  year={2022}
}

@inproceedings{xu2024rethinking,
  title={Rethinking cross-domain sequential recommendation under open-world assumptions},
  author={Xu, Wujiang and Wu, Qitian and Wang, Runzhong and Ha, Mingming and Ma, Qiongxu and Chen, Linxun and Han, Bing and Yan, Junchi},
  booktitle={Proceedings of the ACM Web Conference 2024},
  pages={3173--3184},
  year={2024}
}

@inproceedings{xu2024towards,
  title={Towards open-world cross-domain sequential recommendation: A model-agnostic contrastive denoising approach},
  author={Xu, Wujiang and Ning, Xuying and Lin, Wenfang and Ha, Mingming and Ma, Qiongxu and Liang, Qianqiao and Tao, Xuewen and Chen, Linxun and Han, Bing and Luo, Minnan},
  booktitle={Joint European Conference on Machine Learning and Knowledge Discovery in Databases},
  pages={161--179},
  year={2024},
  organization={Springer}
}

@article{li2024cross,
  title={Cross-Domain Sequential Recommendation via Neural Process},
  author={Li, Haipeng and Cao, Jiangxia and Gao, Yiwen and Liu, Yunhuai and Pang, Shuchao},
  journal={arXiv preprint arXiv:2410.13588},
  year={2024}
}

@inproceedings{liu2025bridge,
  title={Bridge the domains: Large language models enhanced cross-domain sequential recommendation},
  author={Liu, Qidong and Zhao, Xiangyu and Wang, Yejing and Zhang, Zijian and Zhong, Howard and Chen, Chong and Li, Xiang and Huang, Wei and Tian, Feng},
  booktitle={Proceedings of the 48th International ACM SIGIR Conference on Research and Development in Information Retrieval},
  pages={1582--1592},
  year={2025}
}

@article{shen2024exploring,
  title={Exploring user retrieval integration towards large language models for cross-domain sequential recommendation},
  author={Shen, Tingjia and Wang, Hao and Zhang, Jiaqing and Zhao, Sirui and Li, Liangyue and Chen, Zulong and Lian, Defu and Chen, Enhong},
  journal={arXiv preprint arXiv:2406.03085},
  year={2024}
}

@article{wang2025lecdsr,
  title={LeCDSR: Large language model enhanced cross-domain sequential recommendation},
  author={Wang, Shuliang and Zhu, Jiabao and Wang, Kaibo and Ruan, Sijie},
  journal={Information Fusion},
  pages={103762},
  year={2025},
  publisher={Elsevier}
}

@article{xin2025llmcdsr,
  title={Llmcdsr: Enhancing cross-domain sequential recommendation with large language models},
  author={Xin, Haoran and Sun, Ying and Wang, Chao and Xiong, Hui},
  journal={ACM Transactions on Information Systems},
  year={2025},
  publisher={ACM New York, NY}
}

@article{cui2022m6,
  title={M6-rec: Generative pretrained language models are open-ended recommender systems},
  author={Cui, Zeyu and Ma, Jianxin and Zhou, Chang and Zhou, Jingren and Yang, Hongxia},
  journal={arXiv preprint arXiv:2205.08084},
  year={2022}
}

@inproceedings{geng2022recommendation,
  title={Recommendation as language processing (rlp): A unified pretrain, personalized prompt \& predict paradigm (p5)},
  author={Geng, Shijie and Liu, Shuchang and Fu, Zuohui and Ge, Yingqiang and Zhang, Yongfeng},
  booktitle={Proceedings of the 16th ACM conference on recommender systems},
  pages={299--315},
  year={2022}
}

@inproceedings{liu2025uncovering,
  title={Uncovering cross-domain recommendation ability of large language models},
  author={Liu, Xinyi and Wang, Ruijie and Sun, Dachun and Hakkani Tur, Dilek and Abdelzaher, Tarek},
  booktitle={Companion Proceedings of the ACM on Web Conference 2025},
  pages={2736--2743},
  year={2025}
}

@article{peng2024ecellm,
  title={ecellm: Generalizing large language models for e-commerce from large-scale, high-quality instruction data},
  author={Peng, Bo and Ling, Xinyi and Chen, Ziru and Sun, Huan and Ning, Xia},
  journal={arXiv preprint arXiv:2402.08831},
  year={2024}
}

@article{tang2025one,
  title={One model for all: Large language models are domain-agnostic recommendation systems},
  author={Tang, Zuoli and Huan, Zhaoxin and Li, Zihao and Zhang, Xiaolu and Hu, Jun and Fu, Chilin and Zhou, Jun and Zou, Lixin and Li, Chenliang},
  journal={ACM Transactions on Information Systems},
  volume={43},
  number={5},
  pages={1--27},
  year={2025},
  publisher={ACM New York, NY}
}

@article{hou2025weaverec,
  title={WeaveRec: An LLM-Based Cross-Domain Sequential Recommendation Framework with Model Merging},
  author={Hou, Min and Liu, Xin and Wu, Le and He, Chenyi and Liu, Hao and Li, Zhi and Li, Xin and Wei, Si},
  journal={arXiv preprint arXiv:2510.26546},
  year={2025}
}

@inproceedings{hadad2025x,
  title={X-Cross: Dynamic Integration of Language Models for Cross-Domain Sequential Recommendation},
  author={Hadad, Guy and Roitman, Haggai and Eshel, Yotam and Shapira, Bracha and Rokach, Lior},
  booktitle={Proceedings of the 48th International ACM SIGIR Conference on Research and Development in Information Retrieval},
  pages={1497--1507},
  year={2025}
}

@article{yang2408model,
  title={Model merging in llms, mllms, and beyond: Methods, theories, applications and opportunities, 2024},
  author={Yang, Enneng and Shen, Li and Guo, Guibing and Wang, Xingwei and Cao, Xiaochun and Zhang, Jie and Tao, Dacheng},
  journal={URL https://arxiv. org/abs/2408.07666},
  volume={2408},
  number={3}
}

@article{sun2021parallel,
  title={Parallel split-join networks for shared account cross-domain sequential recommendations},
  author={Sun, Wenchao and Ma, Muyang and Ren, Pengjie and Lin, Yujie and Chen, Zhumin and Ren, Zhaochun and Ma, Jun and De Rijke, Maarten},
  journal={IEEE Transactions on Knowledge and Data Engineering},
  volume={35},
  number={4},
  pages={4106--4123},
  year={2021},
  publisher={IEEE}
}

@inproceedings{li2021dual,
  title={Dual attentive sequential learning for cross-domain click-through rate prediction},
  author={Li, Pan and Jiang, Zhichao and Que, Maofei and Hu, Yao and Tuzhilin, Alexander},
  booktitle={Proceedings of the 27th ACM SIGKDD conference on knowledge discovery \& data mining},
  pages={3172--3180},
  year={2021}
}

@article{ma2022mixed,
  title={Mixed information flow for cross-domain sequential recommendations},
  author={Ma, Muyang and Ren, Pengjie and Chen, Zhumin and Ren, Zhaochun and Zhao, Lifan and Liu, Peiyu and Ma, Jun and de Rijke, Maarten},
  journal={ACM Transactions on Knowledge Discovery from Data (TKDD)},
  volume={16},
  number={4},
  pages={1--32},
  year={2022},
  publisher={ACM New York, NY}
}

@article{guo2021gcn,
  title={DA-GCN: A domain-aware attentive graph convolution network for shared-account cross-domain sequential recommendation},
  author={Guo, Lei and Tang, Li and Chen, Tong and Zhu, Lei and Nguyen, Quoc Viet Hung and Yin, Hongzhi},
  journal={arXiv preprint arXiv:2105.03300},
  year={2021}
}

@inproceedings{gao2023autotransfer,
  title={AutoTransfer: Instance transfer for cross-domain recommendations},
  author={Gao, Jingtong and Zhao, Xiangyu and Chen, Bo and Yan, Fan and Guo, Huifeng and Tang, Ruiming},
  booktitle={Proceedings of the 46th international ACM SIGIR conference on research and development in information retrieval},
  pages={1478--1487},
  year={2023}
}

@inproceedings{zhang2024m3oe,
  title={M3oe: Multi-domain multi-task mixture-of experts recommendation framework},
  author={Zhang, Zijian and Liu, Shuchang and Yu, Jiaao and Cai, Qingpeng and Zhao, Xiangyu and Zhang, Chunxu and Liu, Ziru and Liu, Qidong and Zhao, Hongwei and Hu, Lantao and others},
  booktitle={Proceedings of the 47th International ACM SIGIR Conference on Research and Development in Information Retrieval},
  pages={893--902},
  year={2024}
}

@inproceedings{wang2023plate,
  title={PLATE: A prompt-enhanced paradigm for multi-scenario recommendations},
  author={Wang, Yuhao and Zhao, Xiangyu and Chen, Bo and Liu, Qidong and Guo, Huifeng and Liu, Huanshuo and Wang, Yichao and Zhang, Rui and Tang, Ruiming},
  booktitle={Proceedings of the 46th International ACM SIGIR Conference on Research and Development in Information Retrieval},
  pages={1498--1507},
  year={2023}
}

@inproceedings{liu2023diffusion,
  title={Diffusion augmentation for sequential recommendation},
  author={Liu, Qidong and Yan, Fan and Zhao, Xiangyu and Du, Zhaocheng and Guo, Huifeng and Tang, Ruiming and Tian, Feng},
  booktitle={Proceedings of the 32nd ACM International conference on information and knowledge management},
  pages={1576--1586},
  year={2023}
}

@inproceedings{li2023strec,
  title={STRec: Sparse transformer for sequential recommendations},
  author={Li, Chengxi and Wang, Yejing and Liu, Qidong and Zhao, Xiangyu and Wang, Wanyu and Wang, Yiqi and Zou, Lixin and Fan, Wenqi and Li, Qing},
  booktitle={Proceedings of the 17th ACM conference on recommender systems},
  pages={101--111},
  year={2023}
}

@article{ma2024triple,
  title={Triple sequence learning for cross-domain recommendation},
  author={Ma, Haokai and Xie, Ruobing and Meng, Lei and Chen, Xin and Zhang, Xu and Lin, Leyu and Zhou, Jie},
  journal={ACM Transactions on Information Systems},
  volume={42},
  number={4},
  pages={1--29},
  year={2024},
  publisher={ACM New York, NY}
}

@inproceedings{hou2024cross,
  title={Cross-Domain LifeLong Sequential Modeling for Online Click-Through Rate Prediction},
  author={Hou, Ruijie and Yang, Zhaoyang and Ming, Yu and Lu, Hongyu and Zheng, Zhuobin and Chen, Yu and Zeng, Qinsong and Chen, Ming},
  booktitle={Proceedings of the 30th ACM SIGKDD Conference on Knowledge Discovery and Data Mining},
  pages={5116--5125},
  year={2024}
}

@inproceedings{lin2024bridging,
  title={Bridging items and language: A transition paradigm for large language model-based recommendation},
  author={Lin, Xinyu and Wang, Wenjie and Li, Yongqi and Feng, Fuli and Ng, See-Kiong and Chua, Tat-Seng},
  booktitle={Proceedings of the 30th ACM SIGKDD Conference on Knowledge Discovery and Data Mining},
  pages={1816--1826},
  year={2024}
}

@article{xu2025multi,
  title={A multi-view graph contrastive learning framework for cross-domain sequential recommendation},
  author={Xu, Zitao and Chen, Shu and Pan, Weike and Ming, Zhong},
  journal={ACM Transactions on Recommender Systems},
  volume={3},
  number={4},
  pages={1--28},
  year={2025},
  publisher={ACM New York, NY}
}

@article{zang2023contrastive,
  title={Contrastive multi-view interest learning for cross-domain sequential recommendation},
  author={Zang, Tianzi and Zhu, Yanmin and Zhang, Ruohan and Wang, Chunyang and Wang, Ke and Yu, Jiadi},
  journal={ACM Transactions on Information Systems},
  volume={42},
  number={3},
  pages={1--30},
  year={2023},
  publisher={ACM New York, NY}
}

@article{chen2024survey,
  title={A survey on cross-domain sequential recommendation},
  author={Chen, Shu and Xu, Zitao and Pan, Weike and Yang, Qiang and Ming, Zhong},
  journal={arXiv preprint arXiv:2401.04971},
  year={2024}
}

@inproceedings{zheng2025collabedit,
  title={CollabEdit: Towards Non-destructive Collaborative Knowledge Editing},
  author={Zheng, Jiamu and Zhang, Jinghuai and Du, Tianyu and Zhang, Xuhong and Yin, Jianwei and Lin, Tao},
  booktitle={International Conference on Learning Representations},
  volume={2025},
  pages={29154--29174},
  year={2025}
}

@inproceedings{xu-etal-2025-videoeraser,
    title = "{V}ideo{E}raser: Concept Erasure in Text-to-Video Diffusion Models",
    author = "Xu, Naen  and
      Zhang, Jinghuai  and
      Li, Changjiang  and
      Chen, Zhi  and
      Zhou, Chunyi  and
      Li, Qingming  and
      Du, Tianyu  and
      Ji, Shouling",
    booktitle = "Proceedings of the 2025 Conference on Empirical Methods in Natural Language Processing",
    year = "2025",
    pages = "5954--5983",
}

@article{feng2026generalized,
  title={Generalized Category Discovery under Domain Shift: A Frequency Domain Perspective},
  author={Feng, Wei and Ge, Zongyuan},
  journal={Advances in Neural Information Processing Systems},
  volume={38},
  pages={111721--111749},
  year={2025}
}

@ARTICLE{11494811,
  author={Dong, Jiabao and Yang, Lingyuan and Fang, Pengji and Li, Shixiang and Kong, Yusheng and Ren, Lei},
  journal={IEEE Transactions on Automation Science and Engineering}, 
  title={IEI-TIA: Industrial Embodied Intelligence Trustworthy Interpretable Agent for Robotic Long-Horizon and Repetitive Tasks}, 
  year={2026},
  volume={23},
  number={},
  pages={9211-9222},
}

@misc{TEMA,
      title={TEMA: Anchor the Image, Follow the Text for Multi-Modification Composed Image Retrieval}, 
      author={Zixu Li and Yupeng Hu and Zhiheng Fu and Zhiwei Chen and Yongqi Li and Liqiang Nie},
      year={2026},
      eprint={2604.21806},
      archivePrefix={arXiv},
      primaryClass={cs.CV}, 
}

@inproceedings{HABIT,
  title={HABIT: Chrono-Synergia Robust Progressive Learning Framework for Composed Image Retrieval},
  author={Li, Zixu and Hu, Yupeng and Chen, Zhiwei and Zhang, Shiqi and Huang, Qinlei and Fu, Zhiheng and Wei, Yinwei},
  booktitle={Proceedings of the AAAI Conference on Artificial Intelligence},
  volume={40},
  number={8},
  pages={6762--6770},
  year={2026}
}

@InProceedings{pmlr-v267-tan25f,
  title =   {{WM}ark{GPT}: Watermarked Image Understanding via Multimodal Large Language Models},
  author =       {Tan, Songbai and Qiu, Xuerui and Shu, Yao and Xu, Gang and Xu, Linrui and Xu, Xiangyu and Zhuang, Huiping and Li, Ming and Yu, Fei},
  booktitle =   {Proceedings of the 42nd International Conference on Machine Learning},
  pages =   {58621--58636},
  year =   {2025},
   volume =   {267},
  series =   {Proceedings of Machine Learning Research},
  month =   {13--19 Jul},
  publisher =    {PMLR},
  }

@article{gong2020hamming,
  title={Hamming embedding sensitivity guided fusion network for 3D shape representation},
  author={Gong, Biao and Yan, Chenggang and Bai, Junjie and Zou, Changqing and Gao, Yue},
  journal={IEEE Transactions on Image Processing},
  volume={29},
  pages={8381--8390},
  year={2020},
  publisher={IEEE}
}

@inproceedings{fu2025objectrelator,
  title={Objectrelator: Enabling cross-view object relation understanding across ego-centric and exo-centric perspectives},
  author={Fu, Yuqian and Wang, Runze and Ren, Bin and Sun, Guolei and Gong, Biao and Fu, Yanwei and Paudel, Danda Pani and Huang, Xuanjing and Van Gool, Luc},
  booktitle={Proceedings of the IEEE/CVF International Conference on Computer Vision},
  pages={6530--6540},
  year={2025}
}

@inproceedings{tan2025mimir,
  title={Mimir: Improving video diffusion models for precise text understanding},
  author={Tan, Shuai and Gong, Biao and Feng, Yutong and Zheng, Kecheng and Zheng, Dandan and Shi, Shuwei and Shen, Yujun and Chen, Jingdong and Yang, Ming},
  booktitle={Proceedings of the Computer Vision and Pattern Recognition Conference},
  pages={23978--23988},
  year={2025}
}

@inproceedings{tananimate,
  title={Animate-X: Universal Character Image Animation with Enhanced Motion Representation},
  author={Tan, Shuai and Gong, Biao and Wang, Xiang and Zhang, Shiwei and Zheng, DanDan and Zheng, Ruobing and Zheng, Kecheng and Chen, Jingdong and Yang, Ming},
  booktitle={The Thirteenth International Conference on Learning Representations},
  year={2025}
}

@inproceedings{huang2024troika,
  title={Troika: Multi-path cross-modal traction for compositional zero-shot learning},
  author={Huang, Siteng and Gong, Biao and Feng, Yutong and Zhang, Min and Lv, Yiliang and Wang, Donglin},
  booktitle={Proceedings of the IEEE/CVF Conference on Computer Vision and Pattern Recognition},
  pages={24005--24014},
  year={2024}
}

@inproceedings{du2021cert,
  title={Cert-rnn: Towards certifying the robustness of recurrent neural networks},
  author={Du, Tianyu and Ji, Shouling and Shen, Lujia and Zhang, Yao and Li, Jinfeng and Shi, Jie and Fang, Chengfang and Yin, Jianwei and Beyah, Raheem and Wang, Ting},
  booktitle={Proceedings of the 2021 ACM SIGSAC Conference on Computer and Communications Security},
  year={2021}
}

@article{lee2025mitigating,
  title={Mitigating parameter interference in model merging via sharpness-aware fine-tuning},
  author={Lee, Yeoreum and Jung, Jinwook and Baik, Sungyong},
  journal={arXiv preprint arXiv:2504.14662},
  year={2025}
}

@article{wang2025more,
  title={Why do more experts fail? a theoretical analysis of model merging},
  author={Wang, Zijing and Xu, Xingle and Liu, Yongkang and Zhang, Yiqun and Lin, Peiqin and Feng, Shi and Yang, Xiaocui and Wang, Daling and Sch{\"u}tze, Hinrich},
  journal={arXiv preprint arXiv:2505.21226},
  year={2025}
}

@inproceedings{park2024pacer,
  title={Pacer and runner: Cooperative learning framework between single-and cross-domain sequential recommendation},
  author={Park, Chung and Kim, Taesan and Yoon, Hyungjun and Hong, Junui and Yu, Yelim and Cho, Mincheol and Choi, Minsung and Choo, Jaegul},
  booktitle={Proceedings of the 47th International ACM SIGIR Conference on Research and Development in Information Retrieval},
  pages={2071--2080},
  year={2024}
}
